\documentclass[10pt,conference,letterpaper]{IEEEtran}
\usepackage{graphicx}
\usepackage{ragged2e}
\usepackage{url}
\usepackage{color}
\usepackage{footnote}
\usepackage{bm}
\usepackage{pdfpages}
\usepackage{balance}  % for  \balance command ON LAST PAGE  (only there!)
\usepackage{ifpdf}
\usepackage{amsmath}
\usepackage{amsthm}
\usepackage{tabularx}
\usepackage{multirow}
\usepackage[ruled, vlined, boxed, linesnumbered]{algorithm2e}
\SetAlFnt{\small}

\makeatletter
\newcommand{\nosemic}{\renewcommand{\@endalgocfline}{\relax}}% Drop semi-colon ;
\newcommand{\dosemic}{\renewcommand{\@endalgocfline}{\algocf@endline}}% Reinstate semi-colon ;
\let\oldnl\nl% Store \nl in \oldnl
\newcommand{\nonl}{\renewcommand{\nl}{\let\nl\oldnl}}% Remove line number for one line
\makeatother

\usepackage{setspace}
\usepackage{xcolor}
\usepackage{times}
\usepackage{graphicx}
\usepackage{caption}
\usepackage{tikz}
\usepackage{tkz-euclide}
\usetikzlibrary{shapes,arrows}
\usepackage{pgfplots}
\usepackage{pgfplotstable}

\usepackage{array}
\newcolumntype{C}[1]{>{\centering\arraybackslash}m{#1}}

\usepackage{qtree}
\usepackage{xcolor}
\usetikzlibrary{positioning}
\tikzset{main node/.style={circle,draw,minimum size=0.3cm,inner sep=0pt},}

\definecolor{Gray}{gray}{0.5}

\newtheorem{thm}{Theorem}
\newtheorem{lem}{Lemma}

\newtheorem{defn}{Definition}

\def\id#1{\ensuremath{\mathit{#1}}}

\def\idrm#1{\ensuremath{\mathrm{#1}}}

\title{Fast Computation of Graph Edit Distance}
\author{%
{Xiaoyang Chen{$^{\dag}$}, Hongwei Huo{$^{*\dag}$},
Jun Huan{$^{\ddag}$}, Jeffrey~Scott~Vitter {$^{\S}$}}%

\vspace{1.6mm}\\
\fontsize{10}{10}\selectfont\itshape
$^{\dag}$ {\rm Dept. of Computer Science, Xidian University, chenxyu1991@gmail.com, hwhuo@mail.xidian.edu.cn} \\
$^{\ddag}$ {\rm Dept. of Electrical Engineering and Computer Science, The University of Kansas, jhuan@ku.edu} \\
$^{\S}$ {\rm Dept. of Computer and Information Science, The University of Mississippi,
JSV@OleMiss.edu} \\
}

\begin{document}
\maketitle
\begin{abstract}

The graph edit distance (GED) is a well-established distance measure
widely used in many applications. However, existing methods for the
GED computation suffer from several drawbacks including oversized
search space, huge memory consumption, and lots of expensive backtracking.
In this paper, we present {BSS\_GED}, a novel vertex-based mapping
method for the GED computation. First, we create a small search space
by reducing the number of invalid and redundant mappings involved in
the GED computation. Then, we utilize beam-stack search combined with
two heuristics to efficiently compute GED, achieving a flexible
trade-off between available memory and expensive backtracking. Extensive
experiments demonstrate that {BSS\_GED} is highly efficient for the GED
computation on sparse as well as dense graphs and outperforms the
state-of-the-art GED methods. In addition, we also apply {BSS\_GED} to
the graph similarity search problem and the practical results confirm
its efficiency.
\end{abstract}

\section{Introduction}

Graphs are widely used to model various complex structured data, including
social networks, molecular structures, etc. Due to extensive applications
of graph models, there has been a considerable effort in developing techniques
for effective graph data management and analysis, such as graph matching~\cite{RiesenH2009}
and graph similarity search~\cite{ChenHHJ2016, ZengTWFZ2009,ZhaoXLW2012}.

Among these studies, similarity computation between two graphs is a
core and essential problem. In this paper, we focus on the similarity
measure based on graph edit distance (GED) since it is applicable to virtually
all types of data graphs and can also precisely capture structural
differences. Due to the flexible and error-tolerant characteristics of
GED, it has been successfully applied in many applications,
such as molecular comparison in chemistry~\cite{MarinAD2008}, object
recognition in computer vision~\cite{ConteFSV2004} and graph
clustering~\cite{KellyH2005}.

Given two graphs $G$ and $Q$, the GED between them, denoted by $ged(G,Q)$,
is defined as the minimum cost of an edit path that transforms one graph
to another. Unfortunately, unlike the classical
graph matching problem, such as subgraph isomorphism~\cite{YanYH2004}, the fault
tolerance of GED allows a vertex of one graph to be mapped to any
vertex of the other graph, regardless of their labels and degrees.
As a consequence, the complexity of the GED computation is higher
than that of subgraph isomorphism, which has been proved to be an
NP-hard~\cite{ZengTWFZ2009} problem.

The GED computation is usually carried out by means of
a tree search algorithm which explores the space of all possible
mappings of vertices and edges of comparing graphs. The underlying
search space can be organized as an ordered search tree. Based on the
way of generating successors of nodes in the search tree, existing
methods can be divided into two broad categories: vertex-based and
edge-based mapping methods. When generating successors of a node,
the former extends unmapped vertices of comparing graphs, while the
later extends unmapped edges. {A$^\star$-GED}~\cite{RiesenFB2007, RiesenEB2013}
and {DF-GED}~\cite{AbuRRM2015} are two major vertex-based mapping methods.
{A$^\star$-GED} adopts the best-first search paradigm~\textbf{A$^\star$}~\cite{HartNB1968},
which picks up a partial mapping with the minimum induced edit cost
to extend each time. The first found complete mapping induces the GED
of comparing graphs. However, {DF-GED} carries out a depth-first search,
which quickly reaches a leaf node. The edit cost of a leaf node in fact
is an upper bound of~GED and hence can be used to prune nodes later to
accelerate~the GED computation. Different from the above two
methods, {CSI\_GED}~\cite{GoudaH2016} is a novel edge-based mapping
method based on \emph{common substructure isomorphism}, which
works well for the sparse and distant graphs. Similar to {DF-GED},
{CSI\_GED} also adopts the depth-first search paradigm.

Even though existing methods have achieved promising preliminary results,
they still suffer from several drawbacks. Both {A$^\star$-GED}
and {DF-GED} enumerate all possible mappings between
two graphs. However, among these mappings, some mappings
must not be optimal, called \emph{invalid mappings}, or they induce
the same edit cost, called \emph{redundant mappings}. For invalid
mappings, we do not have to generate them, and for redundant mappings,
we only need to generate one of them so as to avoid redundancy.
The search space of {A$^\star$-GED} and {DF-GED} becomes oversized
as they generate plenty of invalid and redundant mappings.

In addition, for {A$^\star$-GED}, it needs to store enormous
partial mappings, resulting in a huge memory consumption.
In practice, {A$^\star$-GED} cannot compute the GED of graphs with
more than 12 vertices. Though {DF-GED} performing a depth-first
search is efficient in memory, it is easily trapped into a local (i.e., suboptimal)
solution and hence produces lots of expensive backtracking. On the
other hand, for {CSI\_GED}, it adopts the depth-first search paradigm,
and hence also faces the expensive backtracking problem. Besides,
the search space of {CSI\_GED} is exponential with respect to
the number of edges of comparing graphs, making it naturally be unsuitable
for dense graphs.

To solve the above issues, we propose a novel vertex-based mapping method
for the GED computation, named {BSS\_GED}, based on \emph{beam-stack search}~\cite{ZhouH2005}
which has shown an excellent performance in AI
literature. Our contributions in this paper are summarized below.

\begin{itemize}
\item
  We propose a novel method of generating successors of nodes
  in the search tree, which reduces a large number of
  invalid and redundant mappings involved in the GED computation.
  As a result, we create a small search space. Moreover, we also
  give a rigorous theoretical analysis of the search space.

\item
  Incorporating with the beam-stack search paradigm into
  our method to compute GED, we achieve a flexible trade-off
  between available memory and the time overhead of backtracking and
  gain a better performance than the best-first and depth-first
  search paradigms.

\item
  We propose two heuristics to prune the search space, where the first
  heuristic produces tighter lower bound and the second heuristic
  enables to fast search of tighter upper bound.

\item
  We have conducted extensive experiments on both real and synthetic
  datasets. The experimental results show that {BSS\_GED} is highly
  efficient for the GED computation on sparse as well as dense graphs, and
  outperforms the state-of-the-art GED methods.

\item
  In addition, we also extend {BSS\_GED} as a standard graph
  similarity search query method and the practical
  results confirm its efficiency.

\end{itemize}

The rest of this paper is organized as follows: In Section~\ref{sec:Problem},
we introduce the problem definition and then give an overview of the vertex-based
mapping method for the GED computation. In Section~\ref{sec:Reduce}, we
create a small search space by reducing the number of invalid and redundant
mappings involved in the GED computation. In Section~\ref{sec:BSSGED}, we
utilize the beam-stack search paradigm to traverse the search space to compute
GED. In Section~\ref{sec:optimization}, we propose two heuristics to
prune the search space. In Section~\ref{sec:ExtensionBSS}, we
extend {BSS\_GED} as a standard graph similarity search query method.
In Section~\ref{sec:experiments}, we report the experimental results
and our analysis. Finally, we investigate research works related
to this paper in Section~\ref{sec:relatedWorks} and then make
concluding remarks in Section~\ref{sec:conclusion}.

\section{Preliminaries}
\label{sec:Problem}

In this section, we introduce basic notations. For simplicity
in exposition, we only focus on simple undirected graphs without
multi-edges or self-loops.

\subsection{Problem Definition}
\label{subsec:problemDefinition}

Let $\Sigma$ be a set of discrete-valued labels. A labeled graph
is a triplet $G =(V_G, E_G, L)$, where $V_G$ is the set of vertices,
$E_G \subseteq V_G \times V_G$ is the set of edges, $L: V_G \cup E_G \to \Sigma$
is a labeling function which assigns a label to a vertex or an edge.
For a vertex $u$, we use $L(u)$ to denote its label.
Similarly, $L(e(u,v))$ is the label of an edge $e(u,v)$. $\id{\idrm{\Sigma}_{V_G}}
= \{L(u): u \in V_G\}$ and $\id{\idrm{\Sigma}_{E_G}} = \{L(e(u,v)): e(u,v) \in E_G\}$
are the label multisets of~$V_G$ and $E_G$, respectively.
For a graph $G$, $S(G) = (V_G, E_G)$ is its unlabeled version,
i.e., its structure. In this paper, we refer $|V_G|$ to the
size of graph $G$.

\begin{defn}[Subgraph Isomorphism~\cite{YanYH2004}]
\label{def:isomorphism}

Given two graphs $G$ and $Q$, $G$ is subgraph isomorphic to $Q$,
denoted by $G \subseteq Q$, if there exists an injective
function~$\phi:V_G \to V_Q$, such that (1) $\forall u \in V_G$,
$\phi(u) \in V_Q$ and $L(u) = L(\phi(u))$. (2) $\forall e(u, v)\in E_G$,
$e(\phi(u), \phi(v)) \in E_Q$ and $L(e(u, v)) = L(e(\phi(u), \phi(v)))$.
If $G \subseteq Q$ and $Q \subseteq G$, then $G$ and $Q$ are graph
isomorphic to each other, denoted by $G \cong Q$.

\end{defn}

There are six edit operations can be used to transform
one graph to another~\cite{WangWYY2012, ZhengZLWZ2015}:
insert/delete an isolated vertex, insert/delete an edge,
and substitute the label of a vertex or an edge.
Given two graphs $G$ and $Q$, an~\emph{edit path}
$P=\langle p_1, \dots, p_k \rangle$ is a sequence of
edit operations that transforms one graph to another,
such as $G = G^{0} \xrightarrow{p_1}, \dots, \xrightarrow{p_k} G^k \cong Q$.
The edit cost of $P$ is defined as the sum of edit cost
of all operations in $P$, i.e., $\sum_{i=1}^{k}c(p_i)$,
where $c(p_i)$ is the edit cost of the edit operation $p_i$.
In this paper, we focus on the uniform cost model, i.e.,
$c(p_i)=1$ for $\forall i$, thus the edit cost of $P$ is
its length, denoted by $|P|$. For $P$, we call it is \emph{optimal}
if and only if it has the minimum length among all possible
edit paths.

\begin{defn}[Graph Edit Distance]
\label{def:editPath}

Given two graphs $G$ and~$Q$, the graph edit distance between them,
denoted by $ged(G, Q)$, is the length of an optimal edit path that
transforms $G$ to $Q$ (or vice versa).

\end{defn}

\noindent \textbf{Example 1}. In Figure~\ref{Fig:OptimalEditPath},
we show an optimal edit path $P$ that transforms graph $G$ to graph $Q$.
The length of $P$ is 4, where we {delete two} edges $e(u_{1}, u_{2})$
and $e(u_1, u_3)$, substitute the label of vertex $u_1$ with label A
and insert one edge $e(u_1, u_4)$ with label~a.

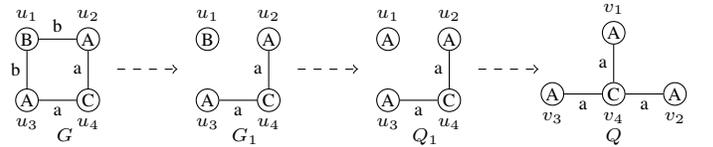
\begin{figure}[!htbp]
\centering
\begin{tikzpicture}[xshift=-1cm]
\tikzset{main node/.style={circle,draw,minimum size=0.3cm,inner sep=0pt},  }
\scriptsize{
    \node[main node] (1) [label=above:$u_1$] {B};
    \node[main node] (2) [right = 0.5cm  of 1,label=above:$u_2$] {A};
    \node[main node] (3) [below = 0.5cm  of 1,label=below:$u_3$] {A};
    \node[main node] (4) [right = 0.5cm  of 3,label=below:$u_4$] {C};

    \path
    (1) edge [above] node {b} (2)
    (1) edge [left]  node {b} (3)
    (2) edge [left]  node {a} (4)
    (3) edge [below] node {a} (4)
    ;

    \draw[->,dashed]  (1.2,-0.4) -- (2.0,-0.4);

    \begin{scope}[xshift=2.4cm]
    \node[main node] (5) [label=above:$u_1$] {B};
    \node[main node] (6) [right = 0.5cm  of 5,label=above:$u_2$] {A};
    \node[main node] (7) [below = 0.5cm  of 5,label=below:$u_3$] {A};
    \node[main node] (8) [right = 0.5cm  of 7,label=below:$u_4$] {C};

    \path[draw,thin]
    (6) edge [left] node {a} (8)
    (7) edge [below] node {a} (8)
    ;
    \end{scope}

    \draw[->,dashed]  (3.6,-0.4) -- (4.4,-0.4);

    \begin{scope}[xshift=4.8cm]
    \node[main node] (1) [label=above:$u_1$] {A};
    \node[main node] (2) [right = 0.5cm  of 1,label=above:$u_2$] {A};
    \node[main node] (3) [below = 0.5cm  of 1,label=below:$u_3$] {A};
    \node[main node] (4) [right = 0.5cm  of 3,label=below:$u_4$] {C};

    \path[draw,thin]
    (2) edge [left] node {a} (4)
    (3) edge [below] node {a} (4)
    ;
    \end{scope}

    \draw[->,dashed]  (6.0,-0.4) -- (6.8,-0.4);

    \begin{scope}[xshift=7.8cm,yshift=-0.73cm]
    \node[main node] (1) [label=below:$v_4$] {C};
    \node[main node] (2) [right = 0.5cm  of 1,label=below:$v_2$] {A};
    \node[main node] (3) [left = 0.5cm  of 1,label=below:$v_3$] {A};
    \node[main node] (4) [above = 0.5cm  of 1,label=above:$v_1$] {A};

    \path[draw,thin]
    (1) edge [left] node {a} (4)
    (1) edge [below] node {a} (2)
    (1) edge [below] node {a} (3)
    ;
    \end{scope}
\tkzText[below](0.5,-1.1){$G$}
\tkzText[below](2.9,-1.1){$G_1$}
\tkzText[below](5.3,-1.1){$Q_1$}
\tkzText[below](7.8,-1.1){$Q$}
}
\end{tikzpicture}
\caption{\small{An optimal edit path $P$ between graphs $G$ and $Q$.}}
\label{Fig:OptimalEditPath}
\end{figure}

\subsection{Graph Mapping}
\label{subsec:GraphMapping}

In this part, we introduce the graph mapping between two graphs,
which can induce an edit path between them. In order to match two unequal
size graphs $G$ and $Q$, we extend their vertex sets
as $V_{G}^{*}$ and $V_{Q}^{*}$ such that $V_{G}^{*} = V_G \cup \{u^n\}$
and $V_{Q}^{*} = V_Q \cup \{v^n\}$, respectively, where $u^n$ and $v^n$
are dummy vertices labeled with $\varepsilon$, s.t., $\varepsilon \notin \Sigma$.
Then, we define graph mapping as follows:

\begin{defn}[Graph Mapping]
\label{def:graphMapping}

A graph mapping from graph~$G$ to graph $Q$ is a bijection
$\psi: V_{G}^{*} \to V_{Q}^{*}$, such that
$\forall u \in V_{G}^{*}$, $\psi(u) \in V_{Q}^{*}$,
and at least one of $u$ and $\psi(u)$ is
not a dummy vertex.

\end{defn}

Given a graph mapping $\psi$ from $G$ to $Q$, it induces an
unlabeled graph $H = (V_H, E_H)$, where $V_H = \{u: u \in V_G \land \psi(u) \in V_Q\}$
and $E_H = \{e(u, v): e(u, v) \in E_G \land e(\psi(u),\psi(v)) \in E_Q\}$,
then $H \subseteq S(G)$ and $H \subseteq S(Q)$.
Let~$G^{\psi}$ (resp. $Q^{\psi}$) be the labeled version of $H$
embedded in $G$ (resp. $Q$). Accordingly, we obtain an edit
path $P_{\psi} : G \to G^{\psi} \to Q^{\psi} \to Q$.
Let $C_D(\psi)$, $C_S(\psi)$ and~$C_I(\psi)$ be the respective edit cost of
transforming $G$ to $G^{\psi}$, $G^{\psi}$ to~$Q^{\psi}$, and~$Q^{\psi}$ to~$Q$.
As $G^{\psi}$ is a subgraph of $G$, we only need to delete vertices and
edges that do not belong to $G^{\psi}$ when transforming $G$ to $G^{\psi}$.
Thus, $C_D(\psi) = |V_G|-|V_H| + |E_G|-|E_H|$.
Similarly, $C_I(\psi)= |V_Q|-|V_H| + |E_Q|-|E_H|$. Since $G^{\psi}$ and
$Q^{\psi}$ have the same structure $H$, we only need to substitute
the corresponding vertex and edge labels between $\id{G^{\psi}}$ and
$\id{Q^{\psi}}$, thus $C_S(\psi) = |\{u: u \in V_H \land L(u) \neq L(\psi(u)) \}|
+ |\{e(u, v):e(u, v) \in E_H \land L(e(u, v)) \neq L(e(\psi(u), \psi(v)))\}|$.

\begin{thm}[\cite{GoudaH2016}]
\label{thm:thm1}

Given a graph mapping $\psi$ from graph $G$ to graph $Q$.
Let $P_{\psi}$ be the edit path induced by~$\psi$, then
$|P_{\psi}| = C_D(\psi) + C_I(\psi) + C_S(\psi)$.

\end{thm}

\noindent \textbf{Example 2.} Consider graphs $G$ and $Q$
in Figure~\ref{Fig:OptimalEditPath}. Given a graph mapping
$\psi:\{u_1, u_2, u_3, u_4\} \to \{v_1, v_2, v_3, v_4\}$,
where $\psi(u_1) = v_1$, $\psi(u_2) = v_2$, $\psi(u_3) = v_3$,
and \mbox{$\psi(u_4) = v_4$}, we have $H = (\{u_1, u_2, u_3, u_4\},
\{e(u_2, u_4), e(u_3, u_4)\})$. Then~$\psi$ induces an edit path
$P_{\psi}: \id{G \to G^{\psi} \to Q^{\psi} \to Q}$ shown in Figure~\ref{Fig:OptimalEditPath},
where $G^{\psi} = G_1$ and $Q^{\psi} = Q_1$. By Theorem~\ref{thm:thm1},
we compute that $C_D(\psi) = 2$, $C_I(\psi) = 1$ and $C_S(\psi) = 1$,
thus $|P_{\psi}| = C_D(\psi) + C_I(\psi) + C_S(\psi) = 4$.

Hereafter, for ease of presentation, we assume that
$G$ and~$Q$ are the two comparing graphs, and
$V_G = \{u_1, \dots, u_{|V_G|}\}$ and $V_Q = \{v_1,\dots,
v_{|V_Q|}\}$. For a graph mapping $\psi$ from $G$ to~$Q$, we
call it is \emph{optimal} only when its induced edit path~$P_{\psi}$
is optimal. Next, we give an overview of the vertex-based mapping
method for computing $ged(G, Q)$ by enumerating all possible graph
mappings from $G$ to $Q$.

\subsection{GED computation: Vertex-based Mapping Approach}
\label{subsec:generalGED}

%Before formally present the method to compute $ged(G, Q)$,
Assuming that vertices in $V_G^{*}$ are processed in the order
$(u_{i_1}, \dots, u_{i_{|V_G|}}, u^n, \dots, u^n)$, where
$i_1, \dots, i_{|V_G|}$ is a permutation of $1, \dots, |V_G|$
detailed in Section~\ref{subsec:orderingVertices}. Then, we denote
a graph mapping from $G$ to $Q$ as
$\psi = \bigcup_{l=1}^{|V_G^{*}|}\{(u_{i_l} \to v_{j_l})\}$ in the
following sections, such that (1) $u_{i_l} = u^n$ if $i_l > |V_G|$;
(2) $v_{j_l} = v^n$ if $j_l > |V_Q|$; and (3) $v_{j_l} = \psi(u_{i_l})$
for $1 \leq l \leq |V_G^{*}|$.

The GED computation is always achieved by means of an ordered
search tree, where inner nodes correspond to partial graph mappings
and leaf nodes correspond to complete graph mappings. Such a search
tree is created dynamically at runtime by iteratively generating
successors linked by edges to the currently considered node.
Let $\id{\psi_r = \{(u_{i_\idrm{1}} \to v_{j_\idrm{1}}), \dots},(u_{i_l} \to v_{j_l})\}$
be the (partial) mapping associated with a node $r$, where
$v_{j_k}$ is the mapped vertex of $u_{i_{k}}$ for $\id{\idrm{1} \leq k \leq l}$,
then Algorithm~\ref{alg:general} outlines the method of generating successors of~$r$.

Algorithm~\ref{alg:general} is easy to understand. First, we
compute the sets of unmapped vertices $C_G^{r}$ and $C_Q^{r}$
in $G$ and $Q$, respectively (line~2). Then, if $|C_{G}^{r}| > 0$,
for the vertex $u_{i_{l+1}}$ to be extended, we choose a vertex $z$
from $C_Q^{r}$ or $\{v^n\}$ as its mapped vertex, and finally
generate all possible successors of $r$ (lines 4--8); otherwise,
all vertices in $G$ were processed, then we insert all
vertices in $C_Q^{r}$ into $G$ and obtain a unique successor
leaf node (lines~10--11).

Staring from a dummy root node $root$ such that $\id{\psi_{root} = \emptyset}$,
we can create the search tree layer-by-layer by iteratively generating
successors. For a leaf node $r$, we compute the edit cost of its corresponding
edit path $P_{\psi_r}$ by Theorem~\ref{thm:thm1}. Thus, when we generate
all leaf nodes, we must find an optimal graph mapping and
then obtain $ged(G, Q)$.

\begin{algorithm}
\caption{{\ttfamily BasicGenSuccr}($r$)}
\label{alg:general}
\SetKwProg{myalg}{Algorithm}{}{}
\DontPrintSemicolon
\SetKwComment{Comment}{$\triangleright$\ }{}
$\psi_r \gets  \{(u_{i_1} \to v_{j_1}), \dots, (u_{i_l} \to v_{j_l})\}, \ succ \gets \emptyset;$\;
$C_{G}^{r} \gets V_G \backslash \{u_{i_1}, \dots, u_{i_l}\},
 C_{Q}^{r} \gets V_Q \backslash \{v_{j_1}, \dots, v_{j_l}\};$\;
\If{$|C_{G}^{r}| > 0 $}
{
    \ForEach{$z \in C_{Q}^{r}$}
    {
        generate successor $q$, s.t., $\psi_q \gets \psi_r \cup \{(u_{i_{l+1}} \to z)\};$\;
        $succ \gets succ \cup \{q\};$\;
    }
    generate successor $q$, s.t., $\psi_q \gets\; \psi_r \cup \{(u_{i_{l+1}} \to v^{n})\};$\;
    $succ \gets succ \cup \{q\};$\;
}
\Else
{
    generate successor $q$, s.t., $\psi_q \gets \psi_r \cup \bigcup_{z \in C_{Q}^{r}}\{(u^n \to z)\};$\;
    $succ \gets succ \cup \{q\};$\;
}
\Return $succ;$\;
\end{algorithm}

However, the above method {\ttfamily BasicGenSuccr} used in {A$^\star$-GED}~\cite{RiesenFB2007}
and {DF-GED}~\cite{AbuRRM2015} generates all possible successors.
As a result, both {A$^\star$-GED} and {DF-GED} enumerate all possible
graph mappings from $G$ to $Q$ and their search space size is $\mathcal{O}(|V_Q|^{|V_G|})$~\cite{GoudaH2016}. However, among these mappings,
some mappings certainly not be optimal, called~\emph{invalid mappings},
or they induce the same edit cost, called~\emph{redundant mappings}.
For invalid mappings, we do not have to generate them, and for redundant
mappings, we only need to generate one of them. Next, we present how to
create a small search space by reducing the number of invalid and
redundant mappings.

\section{Creating Small Search Space}
\label{sec:Reduce}

\subsection{Invalid Mapping Identification}
\label{sec:philen}
Let $|\psi|$ be the length of a graph mapping~$\psi$, i.e., $|V_G^{*}|$.
We give an estimation of~$|\psi|$ in Theorem~\ref{thm:thm2}, which
can be used to identify invalid mappings.

\begin{thm}
\label{thm:thm2}

Given an optimal graph mapping $\psi$ from graph~$G$ to graph $Q$,
then $|\psi| = \idrm{max}\{|V_G|, |V_Q|\}$.

\end{thm}

\begin{proof}
Suppose for the purpose of contradiction that $|\psi| > \idrm{max}\{|V_G|, |V_Q|\}$.
Then $(x \to v^n)$ and $(u^n \to y)$ must be present simultaneously
in $\psi$, where $x \in V_G$ and $y \in V_Q$. We construct another
graph mapping $\psi'= (\psi \backslash \{(x \to v^n),
(u^n \to y)\}) \cup \{(x \to y)\}$, and then prove
$|P_{\psi'}| < |P_{\psi}|$ as follows:

Let $H$ and $H'$ be two unlabeled graphs induced by~$\psi$ and~$\psi'$,
respectively, then $V_{H'}=\{u: u \in V_G \land \psi'(u) \in V_Q\}
= \{u: u \in V_G \land \psi(u) \in V_Q\} \cup \{x: \psi'(x) \in V_Q\}
= V_H \cup \{x\}$. Let $A_x = \{z: z \in V_H \land e(x, z)\in E_G
\land e(y, \psi(z)) \in E_Q\}$,
%be the set of vertices in $V_H$ adjacent to $x$,
then $E_{H'}=\{e(u,v):e(u,v) \in E_G \land e(\psi'(u), \psi'(v)) \in E_Q\}
=\{e(u, v): e(u,v) \in E_G \land e(\psi(u), \psi(v)) \in E_Q\} \cup
\{e(x, z): z \in V_{H'} \land e(x, z) \in E_G \land e(y, \psi(z)) \in E_Q\}
=E_{H} \cup \{e(x, z): z\in A_x\}$. As $x \notin V_H$,
$e(x, z) \notin E_H$ for $\forall z \in A_x$. Thus,
$|V_{H'}| = |V_H| + 1$ and $|E_{H'}| = |E_H| + |A_x|$.

As $C_D(\psi)= |V_G|-|V_H|+|E_G|-|E_H|$ and
$C_I(\psi) = |V_Q| - |V_H| + |E_Q| - |E_H|$, we have
$C_D(\psi') = C_D(\psi)-(1 + |A_x|)$ and $C_I(\psi') = C_I(\psi)-(1 + |A_x|)$.
Since $C_S(\psi)=|\{u:u \in V_H \land L(u) \neq L(\psi(u))\}| +
|\{e(u,v):e(u,v)\in E_H \land L(e(u,v)) \neq L(e(\psi(u),\psi(v)))\}|$,
we have $C_S(\psi') = C_S(\psi) + c(x \to y) + \sum_{z \in A_x}c(e(x, z) \to e(y, \psi(z)))$,
where $c(\cdot)$ gives the edit cost of relabeling a vertex or an edge,
such that $c(a \to b) = 0$ if $L(a) = L(b)$, and $c(a \to b) = 1$ otherwise,
and~$a$ (resp. $b$) is a vertex or an edge in $G$ (resp.~$Q$).
Thus, $C_S(\psi') \leq C_S(\psi) + 1 + |A_x|$. Therefore,
$|P_{\psi'}| = C_D(\psi')+C_I(\psi') + C_S(\psi') \leq
C_D(\psi)-(1 + |A_x|) + C_I(\psi)-(1 + |A_x|) + C_S(\psi)
+ 1 + |A_x| = |P_{\psi}| - (1+|A_x|) < |P_{\psi}|$.
This would be a contradiction that~$P_{\psi}$ is optimal.
Hence $|\psi| = \idrm{max}\{|V_G|, |V_Q|\}$.
\end{proof}

Theorem~\ref{thm:thm2} states that a graph mapping whose length
is greater than $|V|$ must be an invalid mapping, where
$|V| = \idrm{max}\{|V_G|, |V_Q|\}$. For example, considering
graphs $G$ and~$Q$ in Figure~\ref{Fig:OptimalEditPath},
and a graph mapping~$\psi = \{(u_1 \to v_1), (u_2 \to v^n),
(u_3 \to v_3), (u_4 \to v_4), (u^n \to v_2)\}$, we know
that~$\psi$ with an edit cost 7 must be invalid as
$|\psi| =5 > \idrm{max}\{|V_G|, |V_Q|\} = 4$.

\subsection{Redundant Mapping Identification}
\label{sec:redunids}

For a vertex $u$ in $V_Q$, its neighborhood information
is defined as $\id{N_Q(u)= \{(v, L(e(u, v))): v \in V_Q
\land e(u, v) \in E_Q\}}$.

\begin{defn}[Vertex Isomorphism]
\label{def:isoVertices}

Given two vertices \mbox{$u, v \in V_Q$}, $u$ is isomorphic to $v$,
denoted by $u \sim v$, if and only if $L(u) = L(v)$ and $N_Q(u) = N_Q(v)$.

\end{defn}

By Definition~\ref{def:isoVertices}, we know that
the isomorphic relationship between vertices is an equivalence
relation. Thus, we can divide~$V_Q$ into $\lambda_Q$ equivalent classes
$V_{Q}^1, \dots, V_Q^{\lambda_Q}$ of isomorphic vertices. Each
vertex~$u$ is said to belong to class $\pi(u)= i$ if $u \in V_Q^i$.
Dummy vertices in~$\{v^n\}$ are isomorphic to each other, and
let $\pi(v^n) = \lambda_Q + 1$.

\begin{defn}[Canonical Code]
\label{def:canonical}

Given a graph mapping $\psi = \bigcup_{l=1}^{|\psi|}\{(u_{i_l} \to v_{j_l})\}$,
where $v_{j_l} = \psi(u_{i_l})$ for $1 \leq l \leq |\psi|$.
The canonical code of $\psi$ is defined as $code(\psi) = \langle \pi(v_{j_1}),
\dots, \pi(v_{j_{|\psi|}}) \rangle$.

\end{defn}

Given two graph mappings $\psi$ and $\psi'$ such that $|\psi| = |\psi'|$,
we say that $code(\psi) = code(\psi')$ if and only if $\pi(v_{j_l}) = \pi(v_{j_l}')$,
where $v_{j_l} = \psi(u_{i_l})$ and $v_{j_l}'= \psi'(u_{i_l})$ for
$1 \leq l \leq |\psi|$.

\begin{thm}
\label{thm:thm3}

Given two graph mappings $\psi$ and $\psi'$. Let $P_{\psi}$ and
$P_{\psi'}$ be edit paths induced by~$\psi$ and $\psi'$, respectively.
If $code(\psi) = code(\psi')$, then we have $|P_{\psi}| = |P_{\psi'}|$.

\end{thm}

\begin{proof}
As discussed in Section~\ref{subsec:GraphMapping}, $|P_{\psi}|=
C_I(\psi)+ C_D(\psi)+ C_S(\psi)$. In order to prove $|P_{\psi}| = |P_{\psi'}|$,
we first prove $C_I(\psi) = C_I(\psi')$ and $C_D(\psi) = C_D(\psi')$,
then prove $C_S(\psi) = C_S(\psi')$.

Let $H$ and $H'$ be two unlabeled graphs induced by
$\psi$ and~$\psi'$, respectively. For a vertex $u$ in $V_H$,
$\psi(u)$ and $\psi'(u)$ are the mapped vertices of $u$,
respectively. Since $code(\psi) = code(\psi')$, we have
$\pi(\psi(u)) = \pi(\psi'(u))$ and hence obtain
$\psi(u) \sim \psi'(u)$. As $\psi(u) \neq v^{n}$, we have
$\psi'(u) \neq v^{n}$ by Definition~\ref{def:isoVertices}.
Thus $u \in V_{H'}$, and hence we obtain $V_H \subseteq V_{H'}$.
Similarly, we also obtain $V_{H'} \subseteq V_H$. So, $V_H = V_{H'}$.

For an edge $e(u, v)$ in $E_H$, $e(\psi(u), \psi(v))$ is its mapped
edge in $E_Q$. As $\pi(\psi(u)) = \pi(\psi'(u))$, we have
$\psi(u) \sim \psi'(u)$ and then obtain $N_Q(\psi(u)) = N_Q(\psi'(u))$.
Thus, we have $e(\psi'(u), \psi(v)) \in E_Q$.
Similarly, since $\pi(\psi(v)) = \pi(\psi'(v))$, we obtain
$\psi(v) \sim \psi'(v)$. Thus, there must exist edges between
$\psi(u)$ and $\psi'(v)$, $\psi'(u)$ and $\psi'(v)$ (an illustration
is shown in Fig.~\ref{Fig:IllustrateVertices}), thus $e(\psi'(u), \psi'(v)) \in E_Q$
and hence we obtain $e(u, v) \in E_{H'}$. Thus, $E_H \subseteq E_{H'}$.
Similarly, we also obtain $E_{H'} \subseteq E_H$.
So, $E_H = E_{H'}$.

Since $V_H = V_{H'}$ and $E_H = E_{H'}$,
we have $H = H'$. Thus, $C_I(\psi) = C_I(\psi')$ and
$C_D(\psi) = C_D(\psi')$. Next, we do not distinguish
$H$ and $H'$ anymore.

\begin{figure}[!htbp]
\centering
\begin{tikzpicture}[xshift=-1cm]
\tikzset{node/.style={circle, draw, minimum size=0.3cm, inner sep=0pt}, }
\scriptsize{
    \node[draw, circle] at (0, 0)        (1) [label=left:$u$]{};
    \node[draw, circle] at (0, -1.2)     (2) [label=left:$v$]{};
    \path
    (1) edge [] node {} (2)
    ;

    \node[draw, circle] at (2.5, 0)       (3) [label=left:$\psi(u)$]  {};
    \node[draw, circle] at (3.5, 0)       (4) [label=right:$\psi'(u)$] {};
    \node[draw, circle] at (1.84, -1.2)   (5) [label=left:$\psi(v)$]  {};
    \node[draw, circle] at (4.16, -1.2)   (6) [label=right:$\psi'(v)$] {};

    \path[draw,thin]
    (3) edge [thick] node {} (5)
    (3) edge []      node {} (6)
    (4) edge [thick] node {} (5)
    (4) edge [thick] node {} (6)
    ;
    \path[draw, dashdotted]
    (1) edge [->, thick, bend left=20]  node {} (3)
    (1) edge [->, thick, bend left=35]  node {} (4)
    (2) edge [->, thick, bend right=20] node {} (5)
    (2) edge [->, thick, bend right=30] node {} (6)
    ;
}
\end{tikzpicture}
\caption{\small{Illustration of isomorphic vertices.}}
\label{Fig:IllustrateVertices}
\end{figure}
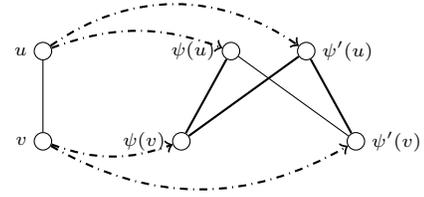

For any vertex $u$ in $V_H$, we have $L(\psi(u)) = L(\psi'(u))$
as $\psi(u) \sim \psi'(u)$. Thus, $|\{u: u \in V_H \land  L(u) \neq L(\psi(u))\}|
= |\{u: u \in V_H \land L(u) \neq L(\psi'(u))\}|$.
For any edge~$e(u, v)$ in $E_H$, since $\psi(u) \sim \psi'(u)$,
we know that $L(e(\psi(u), \psi(v))) = L(e(\psi'(u), \psi(v)))$.
Similarly, we obtain $L(e(\psi'(u), \psi(v))) = L(e(\psi'(u), \psi'(v)))$
as $\psi(v) \sim \psi'(v)$. Thus $L(e(\psi(u), \psi(v))) = L(e(\psi'(u), \psi'(v)))$.
Hence $|\{e(u, v): e(u, v) \in E_H \land L(e(u, v)) \neq L(e(\psi(u), \psi(v)))\}|
=|\{e(u, v): e(u,v) \in E_H \land L(e(u, v)) \neq L(e(\psi'(u), \psi'(v)))\}|$.
Therefore, $C_S(\psi) = C_S(\psi')$. So, we have $|P_\psi| = |P_{\psi'}|$.
\end{proof}

\noindent \textbf{Example 3.} Consider graphs~$G$ and~$Q$
in Figure~\ref{Fig:OptimalEditPath}. For $Q$, we know that
$\id{L(v_\idrm{1}) = L(v_\idrm{2}) = L(v_\idrm{3}) = \idrm{A}}$,
and $N_Q(v_1) = N_Q(v_2) = N_Q(v_3) =\{(v_4, \idrm{a})\}$,
thus $v_1 \sim v_2 \sim v_3$. So, we divide $V_Q$ into
equivalent classes $V_Q^1 = \{v_1, v_2, v_3\}$, $V_Q^2 = \{v_4\}$.
Given two graph mappings $\psi = \{(u_1 \to v_1), (u_2 \to v_2),$ $(u_3 \to v_3),
(u_4 \to v_4)\}$ and $\psi' = \{(u_1 \to v_2), (u_2 \to v_3),$ $(u_3 \to v_1),
(u_4 \to v_4)\}$, we have $code(\psi)= code(\psi') =
\langle 1, 1, 1, 2 \rangle$, and then obtain
$|P_{\psi}| = |P_{\psi'}| = 4$.

Theorem~\ref{thm:thm3} states that graph mappings with the same canonical
code induce the same edit cost. Thus, among these mappings, we only
need to generate one of them. Next, we apply Theorems~\ref{thm:thm2}, \ref{thm:thm3}
into the procedure~{\ttfamily GenSuccr} of generating successors,
which can prevent from generating the above invalid and redundant mappings.

\subsection{Generating Successors}
\label{subsec:generate}

Consider a node $r$ associated with a partial graph mapping
$\psi_r = \{{(u_{i_{1}} \to v_{j_{1}}), \dots, (u_{i_l} \to v_{j_l})}\}$ in the search tree.
Then, the sets of unmapped vertices in $G$ and $Q$ are $C_G^{r} = V_G \backslash \{u_{i_1}, \dots, u_{i_l}\}$ and $C_Q^{r} = V_Q \backslash \{v_{j_1}, \dots, v_{j_l}\}$, respectively.
For the vertex~$u_{i_{l+1}}$ to be extended, let $z \in C_Q^{r} \cup \{v^n\}$ be a possible
mapped vertex.

By Theorem~\ref{thm:thm2}, if $|V_G|\leq |V_Q|$,
we have $|\psi| = |V_Q|$, which means that none of the vertices in $V_G$
is allowed to be mapped to a dummy vertex, i.e., $(u \to v^n) \notin \psi$
for $\forall u \in V_G$.

\textbf{Rule 1}. If $|C_{G}^{r}| \leq |C_{Q}^{r}|$, then $z \in C_{Q}^{r}$;
otherwise $z = v^n$ or $z \in C_{Q}^{r}$.

Applying rule 1 into the process {\ttfamily GenSuccr} of generating
successors of each node, we know that if $|V_G| \leq |V_Q|$
then none of the vertices in $V_G$ will be mapped to a dummy vertex
otherwise only $|V_G|-|V_Q|$ vertices do. As a result, the
obtained graph mapping~$\psi$ must satisfy $|\psi| = \idrm{max}\{|V_G|, |V_Q|\}$.

\begin{defn}[Canonical Code Partial Order]
\label{def:CodeOrder}

Let $\psi$ and $\psi'$ be two graph mappings such that $\id{code(\psi) = code(\psi')}$.
We define that $\psi \preceq \psi'$ if $\exists l, 1 \leq l \leq |\psi|$,
s.t., $\id{\psi(u_{i_k}) = \psi'(u_{i_k})}$ for $1 \leq k < l$
and $\id{\psi(u_\id{i_{l}}) < \psi'(u_{i_{l}})}$.

\end{defn}

By Theorem~\ref{thm:thm3}, we know that graph mappings with the same
canonical code induce the same edit cost, thus among these mappings
we only need to generate the smallest according to the partial order
defined in Definition~\ref{def:CodeOrder}. For $u_{i_{l+1}}$, we only
map it to the smallest unmapped vertex in $V_Q^{m}$, for $\id{\idrm{1} \leq m \leq \lambda_Q}$.
This will guarantee that the obtained graph mapping is smallest among those
mappings with the same canonical code. Then we establish Rule~2 as follows:

\textbf{Rule 2}. $z \in \bigcup_{m=1}^{\lambda_Q}
\idrm{min} \{C_{Q}^{r} \cap V_{Q}^{m}\}$.

Based on the above Rule~1 and Rule~2, we give the method
of generating successors of $r$ in Algorithm~\ref{alg:replacement},
where lines 4--8 corresponds to Rule~2 and lines 9--11 corresponds
to Rule~1.

\begin{algorithm}
\caption{{\ttfamily GenSuccr}($r$)}
\label{alg:replacement}
\DontPrintSemicolon
\SetKwComment{Comment}{$\triangleright$\ }{}
$\psi_r \gets  \{(u_{i_1} \to v_{j_1}), \dots, (u_{i_l} \to v_{j_l})\}, \ succ \gets \emptyset;$\;
$C_{G}^{r} \gets V_G \backslash \{u_{i_1}, \dots, u_{i_l}\},
C_{Q}^{r} \gets V_Q \backslash \{v_{j_1}, \dots, v_{j_l}\};$\;
\If{$|C_{G}^{r}| > 0 $}
{
    \For{$m \gets \idrm{1 \ to} \ \lambda_{Q}$}
    {
        \If{$C_{Q}^{r} \cap V_{Q}^{m} \neq \emptyset$}
        {
            $z \gets \idrm{min}\{C_{Q}^{r} \cap V_{Q}^{m}\};$\;
            generate successor $q$, s.t., $\psi_q \gets \psi_r \cup \{(u_{i_{l+1}} \to z)\};$\;
            $succ \gets succ \cup \{q\};$\;
        }
    }
    \If{$|C_{G}^{r}| > |C_Q^{r}|$}
    {
        generate successor $q$, s.t., $\psi_q \gets \psi_r \cup \{(u_{i_{l+1}} \to v^{n})\};$\;
        $succ \gets succ \cup \{q\};$\;
    }
}
\Else
{
    generate successor $q$, s.t., $\psi_q \gets \psi_r \cup \bigcup_{z \in C_{Q}^{r}}\{(u^n \to z)\};$\;
    $succ \gets succ \cup \{q\};$\;
}
\Return $succ;$\;
\end{algorithm}

\noindent \textbf{Example 4}. Consider graphs $G$ and $Q$ in
Figure~\ref{Fig:OptimalEditPath}. Figure~\ref{Fig:SearchTree} shows
the entire search tree of $G$ and $Q$ created layer-by-layer
by using {\ttfamily GenSuccr}, where vertices in $G$ are processed in
the order $(u_1, u_2, u_3, u_4)$. In a layer, the values inside the
nodes are the possible mapped vertices (e.g., $v_1$ and~$v_4$ in layer one
are the possible mapped vertices of $u_1$). The sequence of vertices
on the path from root to each leaf node gives a complete graph mapping.
In this example, we totally generate~4 graph mappings, and then easily
compute that $ged(G, Q)= 4$.
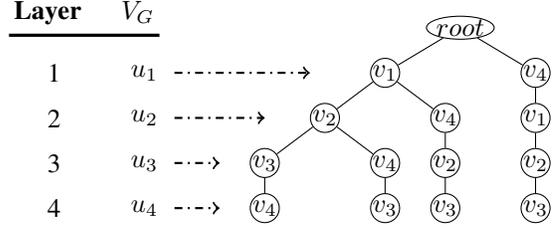
\begin{figure}[!htbp]
\centering
\begin{tikzpicture}[xshift=-1cm]
\tikzset{main node/.style={circle,draw,minimum size=0.3cm,inner sep=0pt},  }

    \draw[very thick](-1, 0.5) -- (1, 0.5);
    \node at (0, 0.8) {\textbf{Layer \ \ \ $V_G$}};

    \node at (-0.4, 0)              ()  {1};
    \node at (-0.4, -0.6)           ()  {2};
    \node at (-0.4, -1.2)           ()  {3};
    \node at (-0.4, -1.8)           ()  {4};

    \node at (0.8, 0)              (u1)  {$u_1$};
    \node at (0.8, -0.6)           (u2)  {$u_2$};
    \node at (0.8, -1.2)           (u3)  {$u_3$};
    \node at (0.8, -1.8)           (u4)  {$u_4$};

    \node[main node, shape=ellipse]  at (5, 0.6)  (root)   {$root$};
    \node[main node]   at (4, 0)       (v11)    {$v_1$};
    \node[main node]   at (6, 0)       (v41)    {$v_4$};
    \node[main node]   at (3.2, -0.6)  (v22)    {$v_2$};
    \node[main node]   at (4.8, -0.6)  (v42)    {$v_4$};
    \node[main node]   at (6, -0.6)    (v12)    {$v_1$};
    \node[main node]   at (2.4, -1.2)  (v33)    {$v_3$};
    \node[main node]   at (4, -1.2)    (v43)    {$v_4$};
    \node[main node]   at (4.8, -1.2)  (v23)    {$v_2$};
    \node[main node]   at (6, -1.2)    (v231)   {$v_2$};
    \node[main node]   at (2.4, -1.8)  (v44)    {$v_4$};
    \node[main node]   at (4, -1.8)    (v341)   {$v_3$};
    \node[main node]   at (4.8, -1.8)  (v342)   {$v_3$};
    \node[main node]   at (6, -1.8)    (v343)    {$v_3$};

    \path[draw,thin]
    (root) edge [] node {} (v11)
    (root) edge [] node {} (v41)
    (v11)  edge [] node {} (v22)
    (v11)  edge [] node {} (v42)
    (v41)  edge [] node {} (v12)
    (v12)  edge [] node {} (v231)
    (v22)  edge [] node {} (v33)
    (v22)  edge [] node {} (v43)
    (v42)  edge [] node {} (v23)
    (v33)  edge [] node {} (v44)
    (v43)  edge [] node {} (v341)
    (v23)  edge [] node {} (v342)
    (v231) edge [] node {} (v343)
    ;
    \draw[->, thick, dashdotted]  (1.2, 0) --    (3.0, 0);
    \draw[->, thick, dashdotted]  (1.2, -0.6) -- (2.4, -0.6);
    \draw[->, thick, dashdotted]  (1.2, -1.2) -- (1.8, -1.2);
    \draw[->, thick, dashdotted]  (1.2, -1.8) -- (1.8, -1.8);

\end{tikzpicture}
\caption{\small{Search tree created by {\ttfamily GenSuccr}.}}
\label{Fig:SearchTree}
\end{figure}

\subsection{Search Space Analysis}
\label{subsec:AnalysisSS}

Replacing {\ttfamily BasicGenSuccr} (see Alg.~\ref{alg:general}) with
{\ttfamily GenSuccr} to generate successors, we reduce
a large number of invalid and redundant mappings and then create a small
search tree. Next, we analyze the size of search tree, i.e.,
the total number of nodes in the search tree.

Nodes in the search tree are grouped into different layers
based on their distances from the root node. Hence, the
search tree is divided into layers, one for each depth.
When all vertices in $V_G$ were processed, for any node
in the layer $|V_G|$ (starting from 0), it generates
a unique successor leaf node, thus we regard this layer as the
last layer. Namely, we only need to generate the first $|V_G|$
layers. For the layer $l$, let $N_{l}$ be the total number of
nodes in this layer. So, the total number of nodes in the search
tree can be computed as $\mathcal{S_R} = \sum_{l=0}^{|V_G|}N_{l}$.

For the layer $l$, the set of vertices in $G$ that have been
processed is $B_{G}^{l} = \{u_{i_1}, \dots, u_{i_l}\}$, correspondingly,
we must choose $l$ vertices from $V_Q \cup \{v^n\}$ as their mapped
vertices. Let $B_{Q}^{l} = \{v_{j_1}, \dots, v_{j_l}\}$ be the $l$ selected vertices,
then we use a vector $\textbf{x} = [x_1, \dots, x_{\lambda_Q +1}]$ to represent it,
where $x_m$ is the number of vertices in $B_{Q}^{l}$ that belong to $V_Q^{m}$,
i.e., $x_m = |B_{Q}^{l} \cap V_Q^{m}|$, for $ 1 \leq m \leq \lambda_Q$,
and $x_{\lambda_Q + 1}$ is the number of dummy vertices in $B_{Q}^{l}$.
Thus, we have
\begin{align}
\label{eq:eq1}
\sum\limits_{m=1}^{\lambda_Q + 1}x_m = l.
\end{align}
where $0 \leq x_m \leq |V_Q^{m}|$ for $1 \leq m \leq \lambda_Q$,
and $x_{\lambda_Q + 1} \geq 0$.

For a solution \textbf{x} of equation~(\ref{eq:eq1}), it corresponds to
a unique~$B_Q^{l}$. The reason is as follows: In Rule~2,
each time we only select the smallest unmapped vertex in $V_Q^{m}$
as the mapped vertex, for $1 \leq m \leq \lambda_{Q}$.
Thus, for $x_m$ in \textbf{x}, it means that~$B_Q^{l}$ contains
the first $x_m$ smallest vertices in $V_Q^{m}$.
For example, let us consider the search tree in Figure~\ref{Fig:SearchTree}.
Let $l = 3$ and $\textbf{x} = [2, 1, 0]$, then
$B_Q^{l}$ contains the first 2 smallest vertices in $V_Q^{1}$, i.e.,
$v_1$ and $v_2$, and the smallest vertex in $V_Q^{2}$, i.e., $v_4$,
then we have $B_Q^{l} = \{v_1, v_2, v_4\}$.

Let $\Psi_{l}$ be the set of solutions of equation~(\ref{eq:eq1}), then
it covers all possible $B_Q^{l}$. For a solution \textbf{x}, it totally
produces $\frac{l!}{\prod_{m=1}^{\lambda_Q + 1}x_m!}$ different (partial)
canonical codes. For example, for $\textbf{x} = [2, 1, 0]$, it produces
3 partial canonical codes, i.e., $\langle 1, 1, 2 \rangle$,
$\langle 1, 2, 1\rangle$ and $\langle 2, 1, 1\rangle$.
As we know, each (partial) canonical code corresponds to
a (partial) mapping from $B_G^{l}$ to $B_Q^{l}$, which is
associated with a node in the layer $l$, thus
\begin{align}
\label{eq:eq2}
N_{l} = \sum\limits_{\textbf{x} \in \Psi_{l}} \frac{l!}{\prod_{m=1}^{\lambda_Q + 1}x_m!}.
\end{align}

In Rule~1, only when the number of unmapped vertices in~$G$ is
greater than that in~$Q$, we select a dummy vertex $v^n$ as the
mapped vertex. As a result, if $|V_G| \leq |V_Q|$ then the number of
dummy vertices in $B_{Q}^{|V_G|}$ is 0 otherwise is $|V_G| - |V_Q|$.
Let $l = |V_G|$ and we then discuss the following two cases:

Case 1. When $|V_G| > |V_Q|$. For any \textbf{x}, we have $x_{\lambda_{Q} + 1} = |V_G| - |V_Q|$.
Then equation~(\ref{eq:eq1}) is reduced to $\sum_{m=\idrm{1}}^{\lambda_Q} x_m = |V_Q|$.
Since $\sum_{m=1}^{\lambda_Q}|V_Q^{m}| = |V_Q|$ and $0 \leq x_m \leq |V_Q^{m}|$,
for $\id{\idrm{1} \leq m \leq \lambda_Q}$, then equation~(\ref{eq:eq1}) has a unique solution
$\textbf{x} = [|V_Q^{1}|, \dots, |V_Q^{\lambda_Q}|, |V_G| - |V_Q|]$,
Thus, $N_{|V_G|} = \frac{|V_G|!}{(|V_G| - |V_Q|)!\prod_{m=\idrm{1}}^{\lambda_Q}|V_Q^{m}|!}$.
As $N_0 = 1$ and $N_{1} \leq \dots \leq N_{|V_G|}$, we obtain
$\mathcal{S_R} \leq |V_G|\frac{|V_G|!}{(|V_G| - |V_Q|)!\prod_{m=1}^{\lambda_Q}|V_Q^{m}|!} + 1$.

Case 2. When $|V_G| \leq |V_Q|$. For any \textbf{x}, we have
$x_{\lambda_{Q}+1} = 0$. Then equation~(\ref{eq:eq1}) is
reduced to $\sum_{m=1}^{\lambda_Q}x_m = |V_G|$.
As $|V_G| \leq |V_Q|$, we have
$N_{|V_G|} = \sum_{\substack{\textbf{x} \in \Psi_{|V_G|}\\ x_{\lambda_{Q} + 1} = 0}}
\frac{|V_G|!}{\prod_{m=1}^{\lambda_Q}x_m!} \leq \sum_{\substack{\textbf{x} \in \Psi_{|V_Q|}\\ x_{\lambda_{Q} + 1} = 0}}\frac{|V_Q|!}{\prod_{m=1}^{\lambda_Q}x_m!}
= \frac{|V_Q|!}{\prod_{m=1}^{\lambda_Q}|V_Q^{m}|!}$.
Since $N_0 = 1$ and $N_{1} \leq \dots \leq N_{|V_G|}$, we have
$\mathcal{S_R} \leq |V_G|\frac{|V_Q|!}{\prod_{m=1}^{\lambda_Q}|V_Q^{m}|!}+1$.
However, this will overestimate $\mathcal{S}_{R}$ when $|V_G| \ll |V_Q|$.

For the layer $l$, if we do not consider the isomorphic vertices
in~$B_Q^{l}$, then there are ${l!}$ mappings from $B_{G}^{l}$ to $B_Q^{l}$.
Since there are at most $\binom{|V_Q|}{l}$ possible~$B_Q^{l}$,
we have $N_{l} \leq \binom{|V_Q|}{l} \cdot{l!}= \frac{|V_Q|!}{(|V_Q|-l)!}$.
So, $\mathcal{S}_R = \sum_{l=0}^{|V_G|}N_{l} \leq
\sum_{l=\idrm{1}}^{|V_G|}\frac{|V_Q|!}{(|V_Q|-l)!} + 1 \\
                 = \frac{|V_Q|!}{(|V_Q|-|V_G|)!} \sum_{l=1}^{|V_G|}
                    \frac{1}{\prod_{m=1}^{|V_G|-l}(|V_Q|-|V_G|+m)} + 1\\
               \leq \frac{|V_Q|!}{(|V_Q|-|V_G|)!}\sum_{l=1}^{|V_G|}\frac{1}{2^{|V_G|-l}} + 1
                    \leq 2\frac{|V_Q|!}{(|V_Q|-|V_G|)!} + 1.$

In summary, if $|V_G|> |V_Q|$, $\mathcal{S}_R = \mathcal{O}(\frac{|V_G||V_G|!}{(|V_G| - |V_Q|)!
\prod_{m=1}^{\lambda_Q}|V_Q^{m}|!})$; otherwise, $\mathcal{S}_R = \mathcal{O}(\idrm{min}\{\frac{|V_G||V_Q|!}{\prod_{m=1}^{\lambda_Q}|V_Q^{m}|!},
\frac{|V_Q|!}{(|V_Q| - |V_G|)!}\})$.

\section{GED Computation using Beam-stack Search}
\label{sec:BSSGED}

The previous section shows that we create a small search space.
However, we still need an efficient search paradigm to traverse
the search space to seek for an optimal graph mapping to compute GED.
In this section, based on the efficient search paradigm, \emph{beam-stack search}~\cite{ZhouH2005},
we give our approach for the GED computation.

\subsection{Data Structures}

For a node $r$ in the search tree, $\id{f(r)=g(r)+h(r)}$ is the
total edit cost assigned to $r$, where $g(r)$ is the edit cost of
the partial path accumulated so far, and $h(r)$ is the estimated
edit cost from $r$ to a leaf node, which is less than or equal to
the real cost. Before formally presenting the algorithm, we first
introduce the data structures used as follows:
\begin{itemize}

\item a beam stack ${bs}$, which is a generalized stack.
Each item in $bs$ is a half-open interval $[f_\id{min}, f_\id{max})$,
and we use~$bs[l]$ to denote the interval of layer $l$.
For a node~$r$ in layer~$l$, its successor $n$ in next
layer $l+1$ is allowed to be expanded only when $f(n)$ is
in the interval $bs[l]$, i.e., $bs[l].f_\id{min} \leq f(n) < bs[l].f_\id{max}$.

\item priority queues $open[0], \dots, open[|V_G|]$,
where $open[l]$ ($0 \leq l \leq |V_G|$) is used to store
the expanded nodes in layer $l$.

\item a table $new$, where $new[\mathcal{H}(r)]$ stores all
successors of~$r$ and $\mathcal{H}(r)$ is a hash function
which assigns a unique ID to~$r$.
\end{itemize}

\subsection{Algorithm}

Algorithm~\ref{alg:BSSGED} performs an iterative search
to obtain a more and more tight upper bound~$ub$ of GED
until $ub = ged(G, Q)$. In an iteration, we perform the
following two steps: (1) we utilize beam search~\cite{RussellN2003}
to quickly reach to a leaf node whose edit cost is an upper bound of GED,
then we update $ub$ (line~4). As beam search expands at most~$w$ nodes
in each layer, some nodes are \emph{inadmissible pruned} when the number
of nodes in a layer is greater than~$w$, where $w$ is the beam width;
Thus,~(2) we backtrack and pop items from $bs$ until
a layer~$l$ such that $bs$.top().$\id{f_\id{max} < ub}$
(lines~5--6), and then shift the range of $bs$.top() (line~9)
to re-expand those inadmissible pruned nodes
in next iteration to search for tighter~$ub$. If $\id{l=-\idrm{1}}$,
it means that we finish a complete search and then
obtain $ub = ged(G,Q)$ (lines~7--8).

In procedure {\ttfamily Search}, we perform a beam search starting from layer
$l$ to re-expand those inadmissible pruned nodes to search for tighter~$ub$,
where $PQL$ and $PQLL$ are two temporary priority queues used to record
expanded nodes in two adjacent layers. Each time we pop a node $r$
with the smallest cost to expand (line~4). If~$r$ is a leaf node,
then we update $ub$ and stop the search as $g(z) \geq g(r)$
holds for $\forall z \in PQL$ (line~7); otherwise, we call
{\ttfamily ExpandNode} to generate all successors of~$r$ that
are allowed to be expanded in next layer and then insert them
into $PQLL$ (lines~\mbox{8--9}). As at most $w$ successors
are allowed to be expanded, we only keep the best $w$ nodes
(i.e., the smallest cost) in $PQLL$, and the nodes left are
inadmissible pruned (lines~11--13). Correspondingly, line~12,
we modify the right boundary of $bs$.top() as the lowest cost
among all inadmissible pruned nodes to ensure that the cost of
the $w$ successors currently expanded is in this interval.

In procedure {\ttfamily ExpandNode}, we generate all successors
of~$r$ that are allowed to be expanded. Note that, all nodes first
generated are marked as false. If~$r$ has not been visited,
i.e., $\id{r.visited = \idrm{false}}$, then we call {\ttfamily GenSuccr}
(i.e., Alg.~\ref{alg:replacement}) to generate all successors of $r$ and
mark $r$ as visited (\mbox{lines 3--4}); otherwise, we directly read
all successors of~$r$ from~$new$ (line~6). For a successor~$n$ of $r$,
if $f(n) \geq ub $ or $n.visited = {\idrm{true}}$, then we safely
prune it, see Lemma~\ref{lem:lem1}. Meanwhile, we delete all successors
of~$n$ from~$new$ (line 9); otherwise, if $bs$.top().$\id{f_\id{min} \leq f(n) < bs}$.top().$\id{f_\id{max}}$, we expand~$n$. If all successors of~$r$
are safely pruned, we safely prune~$r$, and delete~$r$ from
$open[l]$ and its successors from $new$, respectively (line~13).

\begin{lem}
\label{lem:lem1}
In {\ttfamily ExpandNode}, if $f(n) \geq ub $
or $n.visited = {true}$, i.e., line~8,
we safely prune $n$.
\end{lem}
\begin{proof}
For the case $f(n) \geq ub$, it is trivial. Next
we prove it in the other case.

Consider $bs$ in the last iteration. Assuming that in this iteration
we perform {\ttfamily Search} starting from layer $k$ (i.e., backtracking to
layer $k$ in the last iteration, see lines 5--6 in Alg.~\ref{alg:BSSGED}),
and node $r$ and its successors $n$ are in layers~$l$ and~$l+1$,
respectively, then $k \leq l$ and $\id{bs[m].f_{max} \geq ub}$ for
$\id{l + \idrm{1} \leq m \leq |V_G|}$.
If $n.visited = \idrm{true}$, then we must have called {\ttfamily ExpandNode}
to generate successors of $n$ in the last iteration.
For a successor $x$ of $n$ in layer $l+2$,
if $x$ is inadmissible pruned, then $f(x) \geq bs[l+1].f_{max} \geq ub$,
thus we safely prune $x$; otherwise, we consider a
successor of $x$ and repeat this decision process until
a leaf node $z$. Then, it must satisfy that $f(z) = g(z) \geq ub$.
Thus, none of descendants of~$n$ can produce tighter $ub$. So,
we safely prune it.
\end{proof}

\begin{algorithm}
\renewcommand\baselinestretch{1.0}\selectfont
\caption{{BSS\_GED($G, Q, w$)}}
\label{alg:BSSGED}
\SetKwFunction{search}{Search}
\SetKwFunction{expand}{ExpandNode}
\DontPrintSemicolon
\SetKwComment{Comment}{$\triangleright$\ }{}
$\psi_\id{root} \gets \emptyset$, $bs \gets \emptyset, open[] \gets \emptyset, new[]\gets \emptyset$,
$l\gets 0, ub \gets \infty;$\;
$bs.\idrm{push}([0, ub)), open[0].\idrm{push}(\id{root});$\;
\While{$bs \neq \emptyset$}
{
    {\ttfamily Search}$(l, ub, bs, open, new);$\;
    \While{$bs.\idrm{top}().f_\id{max} \geq ub$}
    {
        $bs.\idrm{pop()}, l \gets l-1;$
    }
    \If{$l = -1$}
    {
       \Return $ub$;
    }
    $bs.\idrm{top()}.f_\id{min} \gets bs.\idrm{top}().f_\id{max}, bs.\idrm{top()}.f_\id{max} \gets ub;$ \;
}
\Return $ub$;\;

\setcounter{AlgoLine}{0}
\nonl\SetKwProg{myproc}{procedure}{}{}
\myproc{\search{$l, ub, bs, open, new$}}
{
    $PQL \gets open[l]$, $PQLL \gets \emptyset;$\;
    \While{$PQL \neq \emptyset \ \idrm{or} \ PQLL \neq \emptyset$}
    {
        \While{$PQL \neq \emptyset$}
        {
            $r \gets \idrm{arg \ min}_{n}\{f(n): n \in PQL\};$\;
            $PQL \gets PQL \backslash \{r\};$\;
            \If{$\psi_{r}$ is a complete graph mapping}
            {
                $ub \gets \idrm{min}\{ub, g(r)\}$, \Return;\;
            }
            $\id{succ} \gets ${\ttfamily ExpandNode}$(r, l, ub, open, new);$\;
            $PQLL \gets PQLL \cup succ;$\;
        }
        \If{$|PQLL| > w$}
        {
            $\id{keepNodes} \gets $ the best $w$ nodes in $PQLL;$\;
            $bs.\idrm{top()}.f_{\id{max}} \gets \idrm{min}\{f(n):n\in PQLL
                \land n \notin \id{keepNodes}\};$\;
            $PQLL \gets \id{keepNodes};$\;
        }
        $open[l+1] \gets PQLL, PQL \gets PQLL;$\;
        $PQLL \gets \emptyset, l \gets l+1, bs.\idrm{push}([0, ub));$\;
    }
}
\setcounter{AlgoLine}{0}
\nonl\SetKwProg{myproc}{procedure}{}{}
\myproc{\expand{$r, l, ub, open, new$}}
{
    $expand \gets \emptyset;$\;
    \If{$r.visited = \idrm{false}$}
    {
        $succ \gets ${\ttfamily GenSuccr}$(r);$\;
        $new[\mathcal{H}(r)] \gets succ, r.visited \gets \idrm{true};$\;
    }
    \Else
    {
        $succ \gets new[\mathcal{H}(r)];$\;
    }
    \ForEach{$n \in succ$}
    {
        \If{$f(n) \geq ub \ \idrm{or} \ n.visited = \idrm{true} $}
        {
            $new[\mathcal{H}(n)] \gets \emptyset;$\;
        }
        \ElseIf{$bs.\idrm{top()}.f_{\id{min}} \leq f(n) < bs.\idrm{top()}.f_{\id{max}}$}
        {
            $expand \gets expand \cup \{n\};$\;
        }
    }
    \If {$\forall n \in succ, f(n) \geq ub \ \idrm{or} \ n.visited = \idrm{true}$}
    {
        $open[l] \gets open[l] \backslash \{r\}, new[\mathcal{H}(r)] \gets \emptyset;$\;
    }

    \Return $expand;$ \;
}
\end{algorithm}

\begin{lem}
\label{lem:lem2}
A node $r$ is visited at most $\mathcal{O}(|V_Q|)$ times.
\end{lem}

\begin{proof}
For a node $r$ in layer~$l$, it generates at most \mbox{$|V_Q| + 1$} successors by
{\ttfamily GenSuccr}. In order to fully generate all successors in layer $l+1$,
we backtrack to this layer at most $|open[l]|\cdot(|V_Q| + 1)/w
\leq |V_Q| + 1$ times as $|open[l]| \leq w$. After that,
when we visit $r$ once again, all successors of $r$ are either
pruned or marked, thus we safely prune them by Lemma~\ref{lem:lem1}.
So, $r$ cannot produce tighter~$ub$ in this iteration and we
safely prune it, i.e., lines 12--13 in {\ttfamily ExpandNode}.
Plus the first time when generating~$r$, we totally visit~$r$ at
most $|V_Q|+3$ times, i.e., $\mathcal{O}(|V_Q|)$.
\end{proof}

\begin{thm}
\label{thm:BSSGED}
Given two graphs $G$ and $Q$, {BSS\_GED} must return
$ged(G, Q)$.
\end{thm}
\begin{proof}
By Lemma~\ref{lem:lem2}, a node is visited at most $\mathcal{O}(|V_Q|)$ times,
thus all nodes are totally visited at most $\mathcal{O}(|V_Q|\mathcal{S}_{R})$
times (see $\mathcal{S}_{R}$ in Section~\ref{subsec:AnalysisSS}), which is finite.
So, {BSS\_GED} always terminates. In {\ttfamily Search}, we always update
$ub = \idrm{min}\{ub, g(r)\}$ each time. Thus, $ub$ becomes more and more
tight. Next, we prove that $ub$ converges to $ged(G, Q)$ when {BSS\_GED}
terminates by contradiction.

Suppose that $ub > ged(G, Q)$. Let $r$ and $n$ be the leaf nodes
whose edit cost is $ub$ and $ged(G, Q)$, respectively. Let~$x$ in
layer $l$ be the common ancestor of~$r$ and~$n$, which is farthest
from $root$. Let $z$ in layer $l+1$ be a successor of $x$,
which is the ancestor of~$n$. Then $f(z) \leq f(n) = ged(G, Q) < ub$.

For $z$, it is not in the path from $root$ to $r$, thus it
must be pruned in an iteration, i.e., $f(z) \geq ub$ or $z.visited = \idrm{false}$,
line 8 in {\ttfamily ExpandNode} (if $z$ has been inadmissible pruned,
we backtrack and shift the range of $bs[l]$ to re-expand it until that~$z$
is pruned or marked). For the case $f(z) \geq ub$, it contradicts
that $f(z) < ub$, and for the other case, we conclude that $f(n) \geq ub$
by using the same analysis in Lemma~\ref{lem:lem1}, which contradicts
that $f(n) < ub$. Thus, $ub = ged(G, Q)$.
\end{proof}

\section{Search Space Pruning}
\label{sec:optimization}

In {BSS\_GED}, for a node $r$, if $f(r) = g(r) + h(r) \geq ub$,
then we safely prune $r$. As $g(r)$ is the irreversible edit cost,
the upper bound $ub$ and lower bound $h(r)$ are the keys to perform
pruning. Here, we give two heuristics to prune the search space
as follows:
(1) proposing an efficient heuristic function to obtain tighter $h(r)$;
(2) ordering vertices in $G$ to enable to fast find of tighter $ub$.

\subsection{Estimating $h(r)$}
\label{subsec:lowerBounds}

Let $P$ be an optimal edit path that transforms~$G$ to $Q$,
then~it contains at least $\idrm{max}\{|V_G|, |V_Q|\} -|\Sigma_\id{V_G} \cap \Sigma_\id{V_Q}|$
edit operations performed on vertices. Next, we only consider the edit
operations in $P$ performed on edges. Assuming that we first
delete~$\gamma_1$ edges to obtain~$G_1$, then insert~$\gamma_2$
edges to obtain~$G_2$, and finally change~$\gamma_3$ edge
labels to obtain~$Q$.

When transforming $G$ to $G_1$ by deleting $\gamma_1$ edges,
we have $\id{\idrm{\Sigma}_{E_{G_\idrm{1}}} \subseteq \idrm{\Sigma}_{E_G}}$,
thus $|\idrm{\Sigma}_{E_G} \cap \idrm{\Sigma}_{E_Q}| \geq |\idrm{\Sigma}_{E_{G_\idrm{1}}} \cap \idrm{\Sigma}_{E_Q}|$. When transforming~$G_1$ to $G_2$ by
inserting $\gamma_2$ edges, for each inserted edge, we no longer change
its label, thus $|\idrm{\Sigma}_\id{E_{G_\idrm{2}}} \cap \idrm{\Sigma}_\id{E_Q}|
= |\idrm{\Sigma}_\id{E_{G_\idrm{1}}} \cap \idrm{\Sigma}_\id{E_Q}| + \gamma_\idrm{2}$.
When transforming~$G_2$ to $Q$ by changing~$\gamma_3$ edge labels,
we need to substitute at least $\id{|\idrm{\Sigma}_{E_{Q}}| - |\idrm{\Sigma}_{E_{G_\idrm{2}}}
\cap \idrm{\Sigma}_{E_Q}|}$ edge labels, thus
$\id{\gamma_\idrm{3} \geq |\idrm{\Sigma}_{E_{Q}}| - |\idrm{\Sigma}_{E_{G_\idrm{2}}}
\cap \idrm{\Sigma}_{E_Q}|}$. So, we have
\begin{align}
\label{eq:lower1}
|\Sigma_\id{E_G} \cap \Sigma_\id{E_Q}| + \gamma_2 + \gamma_3 \geq |E_Q|.
\end{align}

Let $lb(G,Q) = \idrm{max}\{|V_G|, |V_Q|\} -|\Sigma_{V_G} \cap \Sigma_{V_Q}|
+ \sum_{i=1}^{3}\gamma_i$, then $ged(G, Q) \geq lb(G,Q)$. Obviously,
the lower bound $lb(G,Q)$ should be as tight as possible.
In order to achieve this goal, we utilize the degree sequence of a graph.

For a vertex $u$ in $G$, its degree $d_u$ is the number of
edges adjacent to $u$. The degree sequence $\delta_G = [\delta_G[1], \dots, \delta_G[|V_G|]]$
of~$G$ is a permutation of $d_1, \dots, d_{|V_G|}$ such that $\id{\delta_G[i]
\geq \delta_G[j]}$ for $\id{i < j}$. For unequal size $G$ and $Q$,
we extend  $\delta_G$ and $\delta_Q$ as $\delta_G' =
[\delta_G[1], \dots, \delta_G[|V_G|], 0_1, \dots, 0_{|V| - |V_G|}]$
and $\delta_Q' = [\delta_Q[1], \dots, \delta_Q[|V_Q|], 0_1, \dots,
0_{|V|-|V_Q|}]$, resp., where $|V|=\idrm{max}\{|V_G|, |V_Q|\}$.
Let $\Delta_1(G, Q)= \lceil\sum_{\delta_{G}'[i] >
\delta_{Q}'[i]}(\delta_{G}'[i]-\delta_{Q}'[i])/2\rceil$ and
$\Delta_2(G, Q)= \lceil \sum_{\delta_{G}'[i] \leq
\delta_{Q}'[i]}(\delta_{Q}'[i]-\delta_{G}'[i])/2\rceil$,
for $1 \leq i \leq |V|$, then
we give the respective lower bounds of~$\gamma_1$ and~$\gamma_2$ as follows:

\begin{thm}[\cite{ChenHHJ2016}]
\label{thm:degree}
Given two graphs $G$ and $Q$, we have
$\gamma_1 \geq \Delta_1(G, Q)$ and $\gamma_2 \geq \Delta_2(G, Q)$.
\end{thm}
\vspace{-5pt}
Based on inequality~(\ref{eq:lower1}) and Theorem~\ref{thm:degree},
we then establish the following lower bound of GED in Theorem~\ref{thm:lowerbounds}.
\vspace{-5pt}
\begin{thm}
\label{thm:lowerbounds}
Given two graphs $G$ and $Q$, we have $ged(G, Q) \geq LB(G, Q)$, where
$LB(G, Q) =  \idrm{max}\{|V_G|, |V_Q|\} - |\Sigma_{V_G} \cap \Sigma_{V_Q}| +
\idrm{max}\{\Delta_1(G, Q) + \Delta_2(G, Q), \Delta_1(G, Q) + |E_Q| - |\Sigma_{E_G} \cap
\Sigma_{E_Q}|\}$.
\end{thm}

Next, we discuss how to estimate~$h(r)$ based on Theorem~\ref{thm:lowerbounds}.
Let $\psi_r = \{(u_{i_\idrm{1}} \to v_{j_\idrm{1}}), \dots, (u_{i_l} \to v_{j_l})\}$
be the partial mapping associated with $r$, then we divide
$G$ into two parts~$G_{1}^{r}$ and $G_{2}^{r}$, where $G_{1}^{r}$
is the mapped part of $G$, s.t., $V_{G_{1}^{r}} = \{u_{i_1},\dots, u_{i_l}\}$
and $E_{G_{1}^{r}} = \{e(u, v): u, v \in V_{G_1^r} \land e(u, v) \in E_G\}$,
and $G_2^r$ is the unmapped part, s.t., $V_{G_2^r} = V_G \backslash V_{G_1^{r}}$
and $E_{G_2^r} = \{e(u, v): u, v \in V_{G_2^r} \land e(u, v) \in E_G\}$.
Similarly, we also obtain $Q_1^r$ and $Q_2^r$. For $r$, by Theorem~\ref{thm:lowerbounds}
we know that $LB(G_2^r, Q_2^r)$ is lower bound of $ged(G_2^r, Q_2^r)$ and
hence can adopt it as $h(r)$. However, for the potential edit cost
on the edges between $G_1^r$ (resp. $Q_1^r$) and $G_2^r$ (resp. $Q_2^r$),
$LB(G_2^{r}, Q_2^{r})$ has not covered~it.

\begin{defn}[Outer Edge Set]
\label{def:outerEdgeSet}
For a vertex $u$ in $G_1^r$, we define its outer edge set
as $O_u = \{e(u, v): v \in V_{G_2^{r}} \land e(u, v) \in E_G\}$,
which consists of edges adjacent to $u$ that neither belong to
$E_{G_1^{r}}$ nor $E_{G_2^{r}}$.
\end{defn}

Correspondingly, $O_{\psi(u)}$ is the outer edge set of $\psi(u)$.
Note that, if $\psi(u) = v^n$, then $O_{\psi(u)} = \emptyset$.
Thus, $\Sigma_{O_u} = \{L(e(u, v)): e(u,v) \in O_u\}$ is the
label multiset of $O_u$. In order to make $O_u$ and
$O_{\psi(u)}$ have the same label multiset, assuming that we
first delete $\xi_1^{u}$ and then insert $\xi_2^{u}$ edges on~$u$,
and finally substitute $\xi_3^{u}$ labels on the outer edges
adjacent to~$u$. Similar to the previous analysis of obtaining
inequality~(\ref{eq:lower1}), we have
\begin{align}
\label{eq:lower5}
\left \{
\begin {array}{l}
       |O_u| - \xi_1^{u} + \xi_2^{u} = |O_{\psi(u)}|\\
       |\Sigma_{O_{u}} \cap \Sigma_{O_{\psi(u)}}| + \xi_2^{u} + \xi_3^{u} \geq |O_{\psi(u)}| \\
\end{array} \right.
\end{align}

Thus, $\sum_{i=1}^{3}\xi_{i}^{u} \geq |O_{\psi(u)}| - |\Sigma_{O_u}
\cap \Sigma_{O_{\psi(u)}}|$. As $|O_u| + \xi_2^{u} = |O_{\psi(u)}| + \xi_1^{u}$,
we have $\sum_{i=1}^{3}\xi_{i}^{u} \geq |O_{\psi(u)}| - |\Sigma_{O_u} \cap \Sigma_{O_{\psi(u)}}|
+ \xi_1^{u} = |O_{u}| - |\Sigma_{O_u} \cap \Sigma_{O_{\psi(u)}}| +
\xi_2^{u} \geq |O_{u}| - |\Sigma_{O_u} \cap \Sigma_{O_{\psi(u)}}|$.
So, $\sum_{i=1}^{3}\xi_{i}^{u} \geq  \idrm{max}\{|O_u|, |O_{\psi(u)}|\}
- |\Sigma_{O_u} \cap \Sigma_{O_{\psi(u)}}|$.
Adding all vertices in $G_1^r$, we obtain the lower bound $LB_1^{r}$ as follows:
\begin{align}
\label{eq:lower6}
\begin{split}
LB_1^{r}=LB(G_2^r, Q_2^r) + \sum\limits_{u \in V_{G_1^{r}}}
(\idrm{max}\{|O_{u}|, |O_{\psi(u)}|\} \\ - |\Sigma_{O_{u}} \cap \Sigma_{O_{\psi(u)}}|).
\end{split}
\vspace{-10pt}
\end{align}

\begin{defn}[Outer Vertex Set]
\label{def:outerVertexSet}
For a vertex $u$ in $G_1^r$, we define its outer vertex set as
$A_u = \{v: v \in V_{G_{2}^{r}} \land e(u, v) \in E_G\}$,
which consists of vertices in $G_2^r$ adjacent to $u$.
\end{defn}

Correspondingly, $A_{\psi(u)}$ denotes the outer vertex set of~$\psi(u)$.
Note that, if $\psi(u) = v^{n}$, then $A_{\psi(u)} = \emptyset$.
Thus, $A_G^r = \bigcup_{u \in V_{G_1^{r}}}A_u$ denotes the set of vertices
in~$G_2^r$ adjacent to those outer edges between $G_1^{r}$ and $G_2^{r}$.
Similarly, we obtain $A_Q^r = \bigcup_{z \in V_{Q_1^{r}}}A_z$.
If $|A_{G}^{r}| \leq |A_{Q}^{r}|$, then we need to insert
at least $|A_{Q}^{r}|-|A_{G}^{r}|$ outer edges on some vertices in $G_{1}^{r}$, hence
$\sum_{u\in V_{G_1^{r}}}\xi_2^{u} \geq |A_{Q}^{r}|-|A_{G}^{r}|$; otherwise,
$\sum_{u\in V_{G_1^{r}}}\xi_1^{u} \geq |A_{G}^{r}|-|A_{Q}^{r}|$.
Considering equation~(\ref{eq:lower5}), for a vertex $u$ in $G_{1}^{r}$,
we have $\xi_2^{u} + \xi_3^{u} \geq |O_{\psi(u)}| - |\Sigma_{O_u} \cap \Sigma_{O_{\psi(u)}}|$.
Thus, $\sum_{u\in V_{G_1^{r}}}(\xi_1^{u} + \xi_2^{u}
+ \xi_3^{u}) \geq \sum_{u\in V_{G_1^{r}}}(|O_{\psi(u)}| - |\Sigma_{O_u} \cap \Sigma_{O_{\psi(u)}}|)
+ \idrm{max}\{0, |A_{G}^{r}|-|A_{Q}^{r}|\}$. As $|O_u|+ \xi_2^{u}
= |O_{\psi(u)}| + \xi_1^{u}$, \\we have $\sum_{u\in V_{G_1^{r}}}(\xi_1^{u}
+ \xi_2^{u} + \xi_3^{u}) \geq \sum_{u\in V_{G_1^{r}}} (|O_{u}|
- |\Sigma_{O_u} \cap \Sigma_{O_{\psi(u)}}|) + \idrm{max}\{0,
|A_{Q}^{r}|-|A_{G}^{r}|\}$. So, we obtain the lower bounds
$LB_2^r$ and $LB_3^r$ as follows:
\begin{align}
\begin{split}
LB_2^{r} = LB(G_2^r, Q_2^r) + \sum_{u\in V_{G_1^{r}}}(|O_{\psi(u)}| -
|\Sigma_{O_u} \cap \Sigma_{O_{\psi(u)}}|)
\\ + \idrm{max}\{0, |A_{G}^{r}|-|A_{Q}^{r}|\}.
\end{split}\\
\begin{split}
LB_3^{r} = LB(G_2^r, Q_2^r) + \sum_{u\in V_{G_1^{r}}}(|O_{u}| -
|\Sigma_{O_u} \cap \Sigma_{O_{\psi(u)}}|)
\\ + \idrm{max}\{0, |A_{Q}^{r}|-|A_{G}^{r}|\}.
\end{split}
\end{align}

Based on the above lower bounds $LB_1^r$, $LB_2^r$ and $LB_3^r$, we
adopt $h(r) = \idrm{max}\{LB_1^r, LB_2^r, LB_3^r\}$ as the
heuristic function to estimate the edit cost of a node $r$ in {BSS\_GED}.

\noindent \textbf{Example 5.} Consider graphs $G$ and $Q$ in
Figure~\ref{Fig:OutEdges}. For a node $r$ associated with a
partial mapping $\psi(r) = \{(u_1 \to v_1), (u_2 \to v_2)\}$,
then $G_2^{r} = (\{u_3, u_4, u_5\}, \{e(u_3, u_5), e(u_4, u_5)\}, L)$
and $Q_2^{r} = (\{v_3, v_4, v_5, v_6\}, \{e(v_3, v_6), e(v_4, v_6),
e(v_5, v_6)\}, L)$. By Theorem~\ref{thm:lowerbounds}, we compute
$LB(G_2^{r}, Q_2^{r}) = 2$. Considering vertices~$u_1$ and $u_2$ that
have been processed, we have $O_{u_1} = \{e(u_1, u_3), e(u_1, u_4)\}$ and
$O_{u_2} = \{e(u_2, u_4)\}$, and then obtain $\Sigma_{O_{u_1}}
= \{\idrm{a},\idrm{a}\}$ and $\Sigma_{O_{u_2}} = \{\idrm{b}\}$.
Similarly, we have $\Sigma_{O_{v_1}} = \{\idrm{a}, \idrm{a}\}$
and $\Sigma_{O_{v_2}} = \{\idrm{b}\}$. Thus $LB_1^r = LB(G_2^{r}, Q_2^{r})
+\sum_{u \in \{u_1, u_2\}}(\idrm{max}\{|O_{u}|, |O_{\psi(u)}|\} - |\Sigma_{O_{u}}
\cap \Sigma_{O_{\psi(u)}}|) = 2$. As $A_G^{r} = \{u_3, u_4\}$ and
$A_Q^{r} = \{v_3, v_4, v_5\}$, we obtain
$LB_2^r =LB(G_2^{r}, Q_2^{r}) + \sum_{u \in \{u_1, u_2\}}
(|O_{\psi(u)}| - |\Sigma_{O_u} \cap \Sigma_{O_{\psi(u)}}|)
+ \idrm{max}\{0, |A_G^{r}| - |A_Q^{r}|\} = 2$,
and $LB_3^r = LB(G_2^{r}, Q_2^{r}) + \sum_{u \in \{u_1, u_2\}}
(|O_{u}| - |\Sigma_{O_u} \cap \Sigma_{O_{\psi(u)}}|) + \idrm{max}\{0, \id{|A_Q^{r}|
- |A_G^{r}|}\}\\=3$.
So, $h(r) = \idrm{max}\{LB_1^r, LB_2^r, LB_3^r\} = \idrm{max}\{2, 2, 3\} = 3$.

\begin{figure}[htbp]
\centering
\begin{tikzpicture}[xshift=-1cm]
\tikzset{node/.style={circle, draw, minimum size=0.2cm, inner sep=0pt},}
\node[draw, circle, fill=gray!30] at (0.3, 0)        (u1) [label=left:$u_1$]{\scriptsize{A}};
\node[draw, circle, fill=gray!30] at (2.05, 0)       (u2) [label=right:$u_2$]{\scriptsize{B}};
\node[draw, circle]               at (0.3, -0.9)     (u3) [label=left:$u_3$]{\scriptsize{A}};
\node[draw, circle]               at (2.05,-0.9)     (u4) [label=right:$u_4$]{\scriptsize{A}};
\node[draw, circle]               at (1.2, -1.8)     (u5) [label=left:$u_5$]{\scriptsize{C}};
\path
(u1) edge [thick, left] node {a} (u3)
(u1) edge [thick, above] node {a} (u4)
(u2) edge [thick, right] node {b} (u4)
(u3) edge [left] node {a} (u5)
(u4) edge [right] node {b} (u5)
;

\node[draw, circle, fill=gray!30] at (4.39, 0)      (v1) [label=left:$v_1$]{\scriptsize{A}};
\node[draw, circle, fill=gray!30] at (5.6, 0)       (v2) [label=right:$v_2$]{\scriptsize{B}};
\node[draw, circle]               at (3.84, -0.9)   (v3) [label=left:$v_3$]{\scriptsize{A}};
\node[draw, circle]               at (5.1, -0.9)    (v4) [label=left:$v_4$]{\scriptsize{A}};
\node[draw, circle]               at (6.36, -0.9)   (v5) [label=right:$v_5$]{\scriptsize{B}};
\node[draw, circle]               at (4.84, -1.8)   (v6) [label=left:$v_6$]{\scriptsize{C}};

\path[draw, thin]
(v1) edge [thick, left] node {a} (v3)
(v1) edge [thick, right] node {a} (v4)
(v2) edge [thick, right] node {b} (v5)
(v3) edge [left] node {a} (v6)
(v4) edge [left] node {a} (v6)
(v5) edge [right] node {b} (v6)
;

\tkzText[below](1.2,-2.1){\small{$G$}}
\tkzText[below](4.9,-2.1){\small{$Q$}}

\end{tikzpicture}
\caption{\small{Example of two comparing graphs $G$ and $Q$.}}
\label{Fig:OutEdges}
\end{figure}
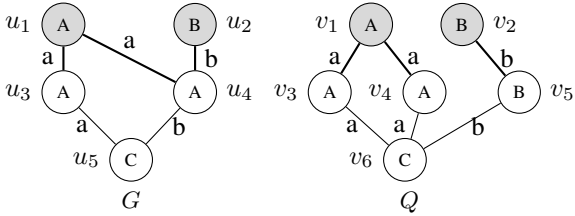

\subsection{Ordering Vertices in $G$}
\label{subsec:orderingVertices}

In {BSS\_GED}, we use {\ttfamily GenSuccr} to generate
successors. However, we need to determine
the processing order of vertices in $G$ at first, i.e., $(u_{i_1}, \dots, u_{i_{|V_G|}})$
(see Section~\ref{subsec:generalGED}). The most primitive
way is to adopt the default vertex order in~$G$,
i.e., $(u_1, \dots, u_{|V_G|})$, which is used in
{A$^\star$-GED}~\cite{RiesenFB2007} and {DF-GED}~\cite{AbuRRM2015}.
However, this may be inefficient as it has not considered
the structure relationship between vertices.

For vertices $u$ and $v$ in $G$ such that $e(u,v) \in E_G$,
if $u$ has been processed, then in order to early obtain the
edit cost on~$e(u,v)$, we should process $v$ as soon as possible.
Hence, our policy is to traverse $G$ in a depth-first order to
obtain $(u_{i_1}, \dots, u_{i_{|V_G|}})$. However, starting from
different vertices to traverse, we may obtain different orders.

In section~\ref{subsec:lowerBounds}, we have proposed the heuristic
estimate function $h(r)$, where an important part is $LB(G_2^{r}, Q_2^{r})$
presented in Theorem~\ref{thm:lowerbounds}. As we know,
the more structure $G_2^{r}$ and~$Q_2^{r}$ keep, the tighter
lower bound $LB(G_2^{r}, Q_2^{r})$ we may obtain. As a result,
we preferentially consider vertices with small degrees. This
is because that when we first process those vertices,
the left unmapped parts $G_2^{r}$ and~$Q_2^{r}$ could keep the
structure as much as possible.

\begin{defn}[Vertex Partial Order]
\label{def:vertexPartialOrder}

For two vertices $u$ and~$v$ in $G$, we define that $u \prec v$
if and only if $d_u < d_v$ or $d_u = d_v \land u < v$.
\end{defn}

In Algorithm~\ref{alg:determineOrder}, we give the method to compute
the order $(u_{i_1}, \dots, u_{i_{|V_G|}})$. First, we sort vertices to obtain
a global order array $rank$ based on the partial order $\prec$ (line~2).
Then, we call {\ttfamily DFS} to traverse $G$ in a depth-first order (lines 3--6).

In {\ttfamily DFS}, we sequentially insert $u$ into $order$ and
then mark~$u$ as visited, i.e., set $F[u] = \idrm{true}$ (line~1).
Then, we obtain the set $N_u$ of vertices adjacent to~$u$ (line~2).
Finally, we select a smallest unvisited vertex $v$ from $N_u$ based on the
partial order $\prec$, and then recursively call {\ttfamily DFS}
to traverse the subtree rooted at $v$ (lines 3--7).

\noindent \textbf{Example 6.} For the graph $G$ in Figure~\ref{Fig:OutEdges},
we first compute $rank = [u_2, u_1, u_3, u_5, u_4]$. Starting from
$u_2$, we traverse $G$ in a depth-first order, and finally
obtain $order=[u_2, u_4, u_1, u_3, u_5]$. Thus, we process
vertices in $G$ in the order $(u_2, u_4, u_1, u_3, u_5)$ in {BSS\_GED}.

\section{Extension of BSS\_GED}
\label{sec:ExtensionBSS}
In this section, we extend {BSS\_GED} to solve the GED based
graph similarity search problem: Given a graph database $\mathcal{G} =
\{\mathcal{G}_1, \mathcal{G}_2, \dots \}$, a query graph $\mathcal{Q}$
and an edit distance threshold $\tau$, the problem aims to find
all graphs in $\mathcal{G}$ satisfy $\id{ged(\mathcal{G}_i, \mathcal{Q}) \leq \tau}$.
As computing GED is an NP-hard problem, most of the existing methods,
such as~\cite{WangWYY2012, ChenHHJ2016, ZhaoXLW2012, ZhengZLWZ2015},
all use the filter-and-verify schema, that is, first filtering
some graphs in~$\mathcal{G}$ to obtain candidate graphs,
and then verifying them.

Here, we also use this strategy. For each data graph $\mathcal{G}_{i}$,
we compute the lower bound $LB(\mathcal{G}_{i}, \mathcal{Q})$
by Theorem~\ref{thm:lowerbounds}. If $LB(\mathcal{G}_{i}, \mathcal{Q}) > \tau$,
then $ged(\mathcal{G}_{i}, \mathcal{Q}) \geq LB(\mathcal{G}_{i}, \mathcal{Q}) > \tau$
and hence we filter $\mathcal{G}_{i}$; otherwise, $\mathcal{G}_{i}$
becomes a candidate graph.

For a candidate graph $\mathcal{G}_{i}$, we need to compute
$ged(\mathcal{G}_i,\mathcal{Q})$ to verify it. The standard
method is that we first compute $ged(\mathcal{G}_i, \mathcal{Q})$
and then determine $\mathcal{G}_{i}$ is a required graph or not
by judging $ged(\mathcal{G}_i, \mathcal{Q}) \leq \tau$.
Incorporating~$\tau$ with {BSS\_GED}, we can further
accelerate it as follows: First, we set the initial
upper bound $ub$ as $\tau + 1$ (line 1 in Alg.~\ref{alg:BSSGED}).
Then, during the execution of {BSS\_GED},
when we reach to a leaf node $r$, if the cost of~$r$ (i.e., $g(r)$)
satisfies $g(r) \leq \tau$, then $\mathcal{G}_{i}$ must be
a required graph and we stop running of {BSS\_GED}.
The reason is that $g(r)$ is an upper bound of GED and hence
we know that $ged(\mathcal{G}_{i}, \mathcal{Q}) \leq g(r) \leq \tau$.
\begin{algorithm}
\caption{{\ttfamily DetermineOrder}$(G)$}
\label{alg:determineOrder}
\SetKwFunction{DFS}{DFS}
\DontPrintSemicolon
\SetKwComment{Comment}{$\triangleright$}{}
$\id{F}[1..|V_G|] \gets \idrm{false}, \id{order}[] \gets \emptyset$, $\id{count} \gets 1;$\;
$\id{rank} \gets$ sort vertices in $V_G$ according to the partial order $\prec$;\;
\For{$i \gets 1 \ \idrm{to} \ |V_G|$}
{
   $u \gets \id{rank}[i];$\;
   \If{$F[u] = \idrm{false}$}
   {
        {\ttfamily DFS$(u, F, rank, order, count)$}\;
   }
}
\Return $\id{order}$\;
\setcounter{AlgoLine}{0}
\nonl\SetKwProg{myproc}{procedure}{}{}
\myproc{\DFS{$u, F, rank, order, count$}}
{
    $order[count] \gets u, count \gets count + 1, F[u] \gets \idrm{true};$\;
    $N_u \gets \{v: v \in V_G \land e(u, v) \in E_G\};$\;
    \While{$|N_u| > 0$}
    {
        $v\gets \idrm{argmin}_{j}\{\id{rank}[j]: j \in N_u\};$\;
        $N_u \gets N_u \backslash\{v\};$\;
        \If{$F[v] = \idrm{false}$}
        {
            {\ttfamily DFS$(v, F, rank, order, count);$}\;
        }
    }
}
\end{algorithm}

\begin{figure*}[htbp]
\centering
    \begin{minipage}{0.28\linewidth}
		\centerline{\includegraphics[width=1\textwidth, height=3cm]{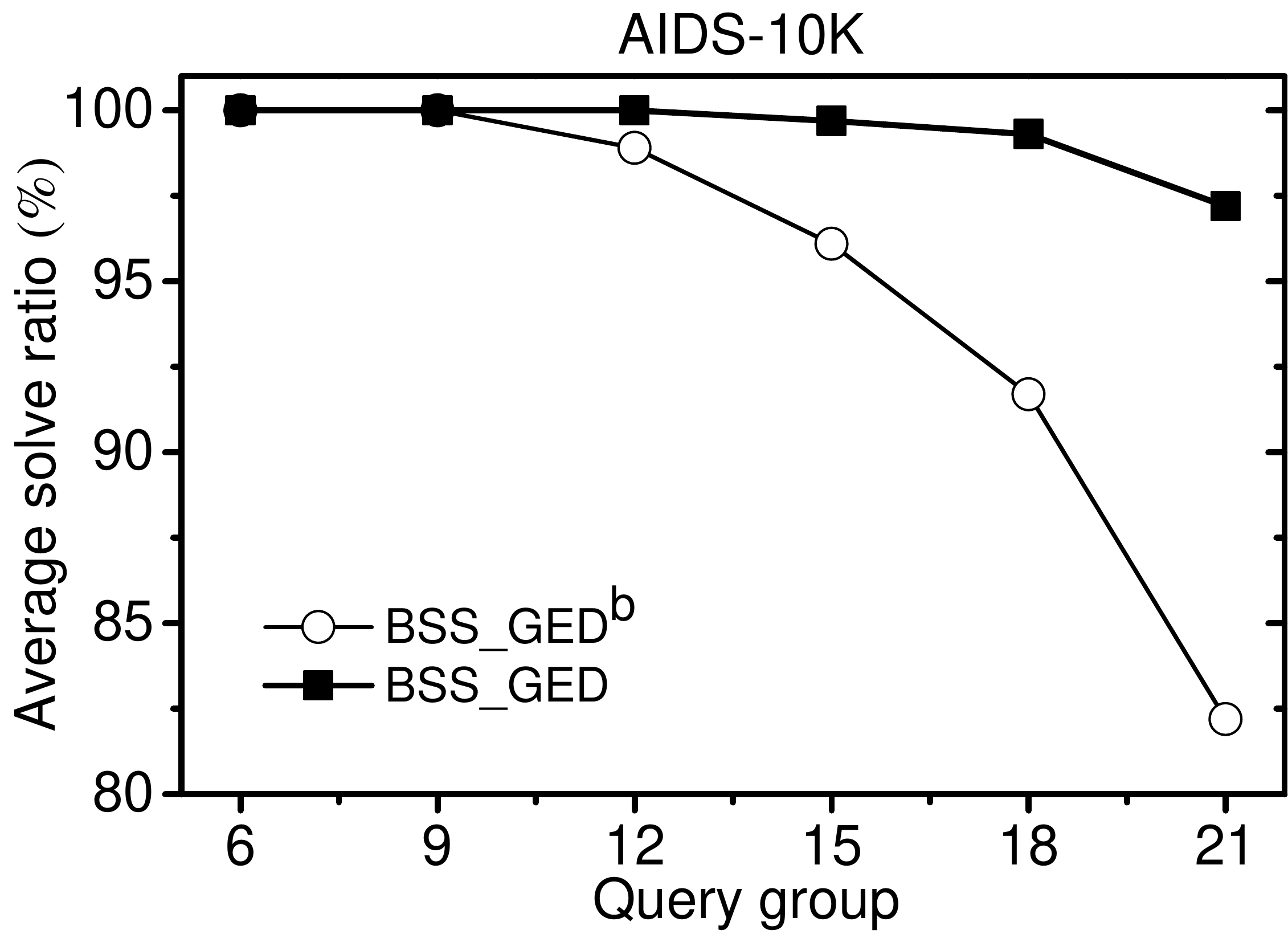}}
	\end{minipage}
	\qquad
	\begin{minipage}{0.28\linewidth}
		\centerline{\includegraphics[width=1\textwidth, height=3cm]{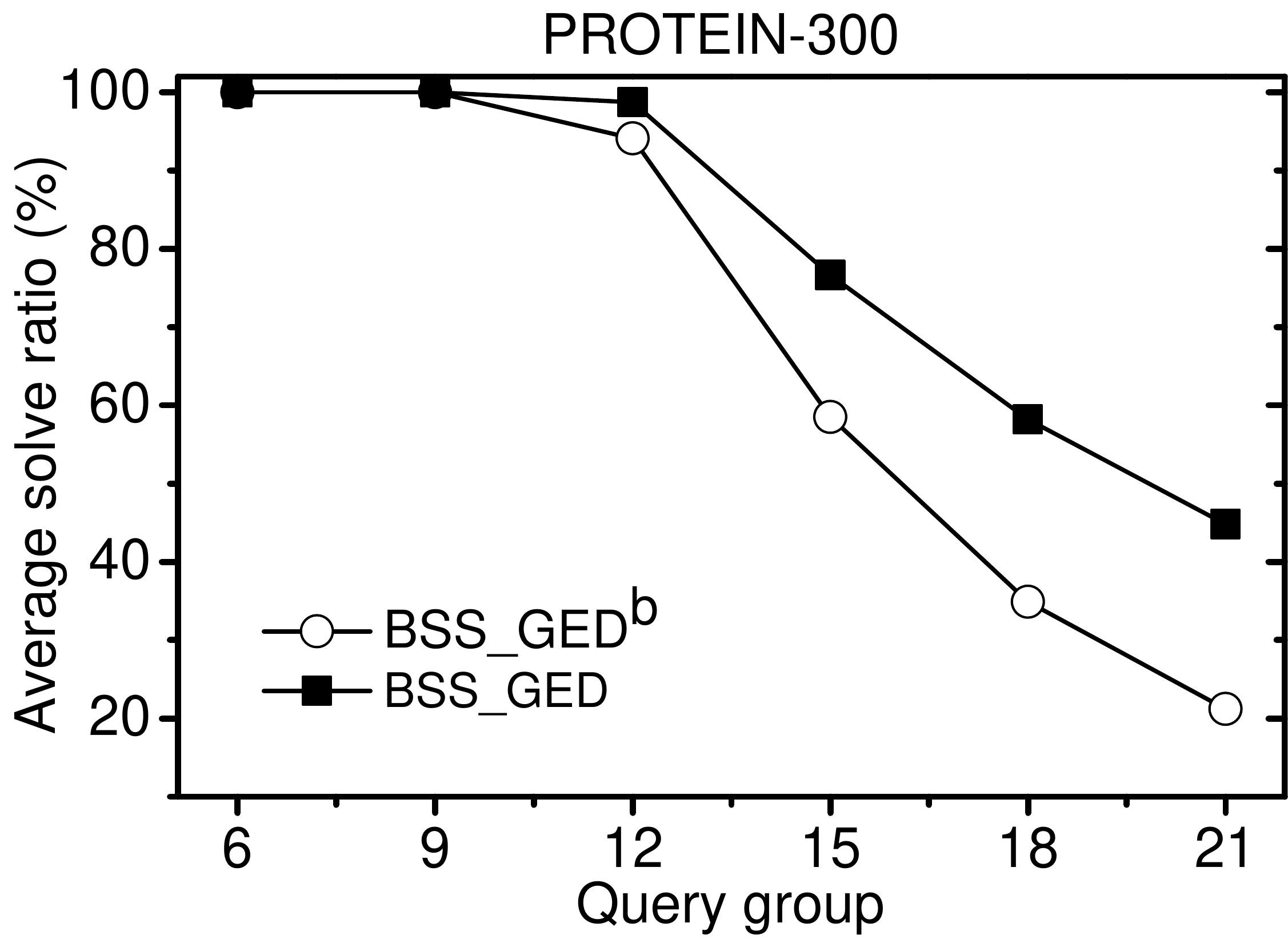}}
	\end{minipage}
	\qquad
	\begin{minipage}{0.28\linewidth}
		\centerline{\includegraphics[width=1\textwidth, height=3cm]{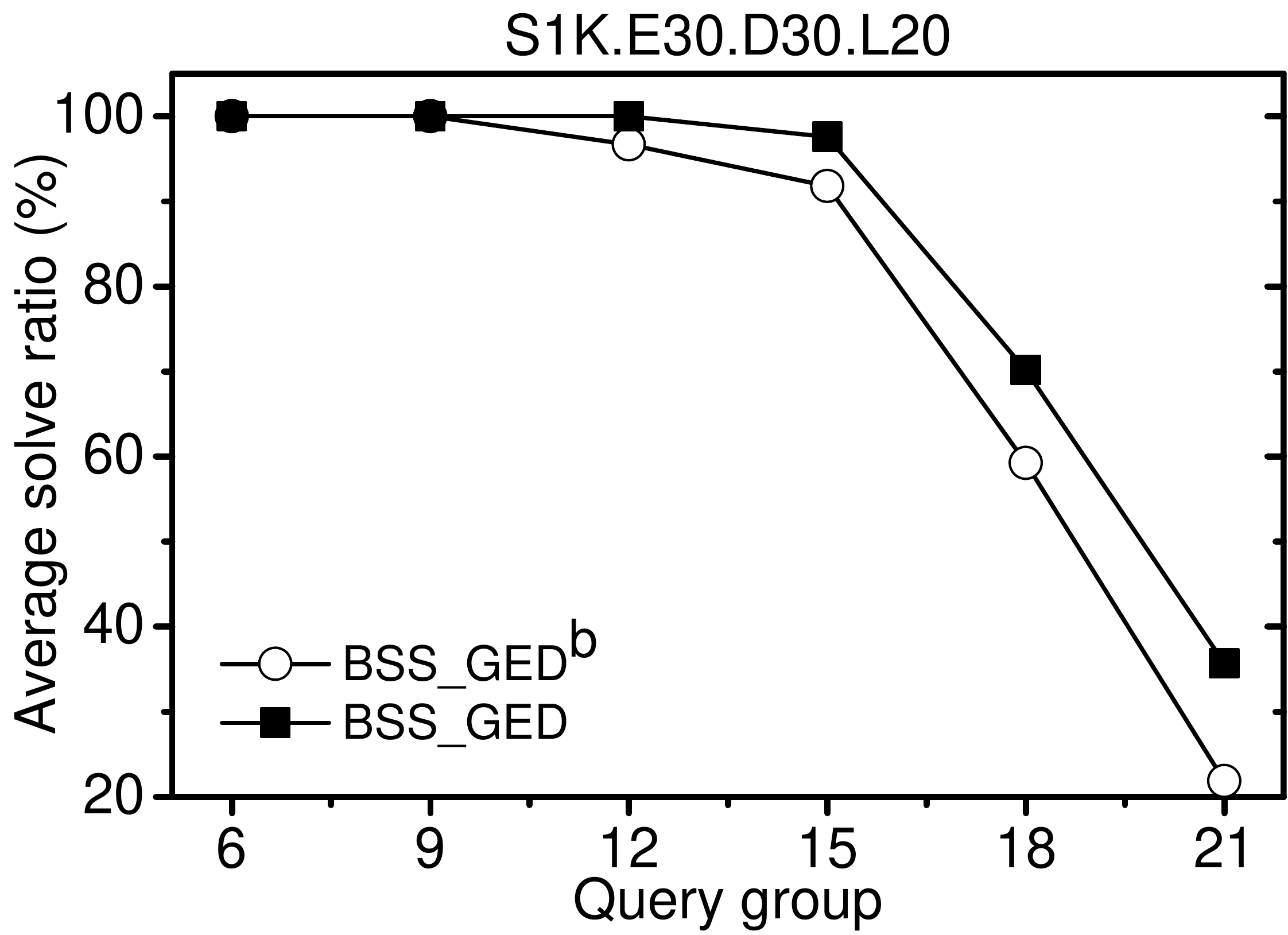}}
    \end{minipage}
    \vfill
    \begin{minipage}{0.28\linewidth}
		\centerline{\includegraphics[width=1\textwidth, height=3cm]{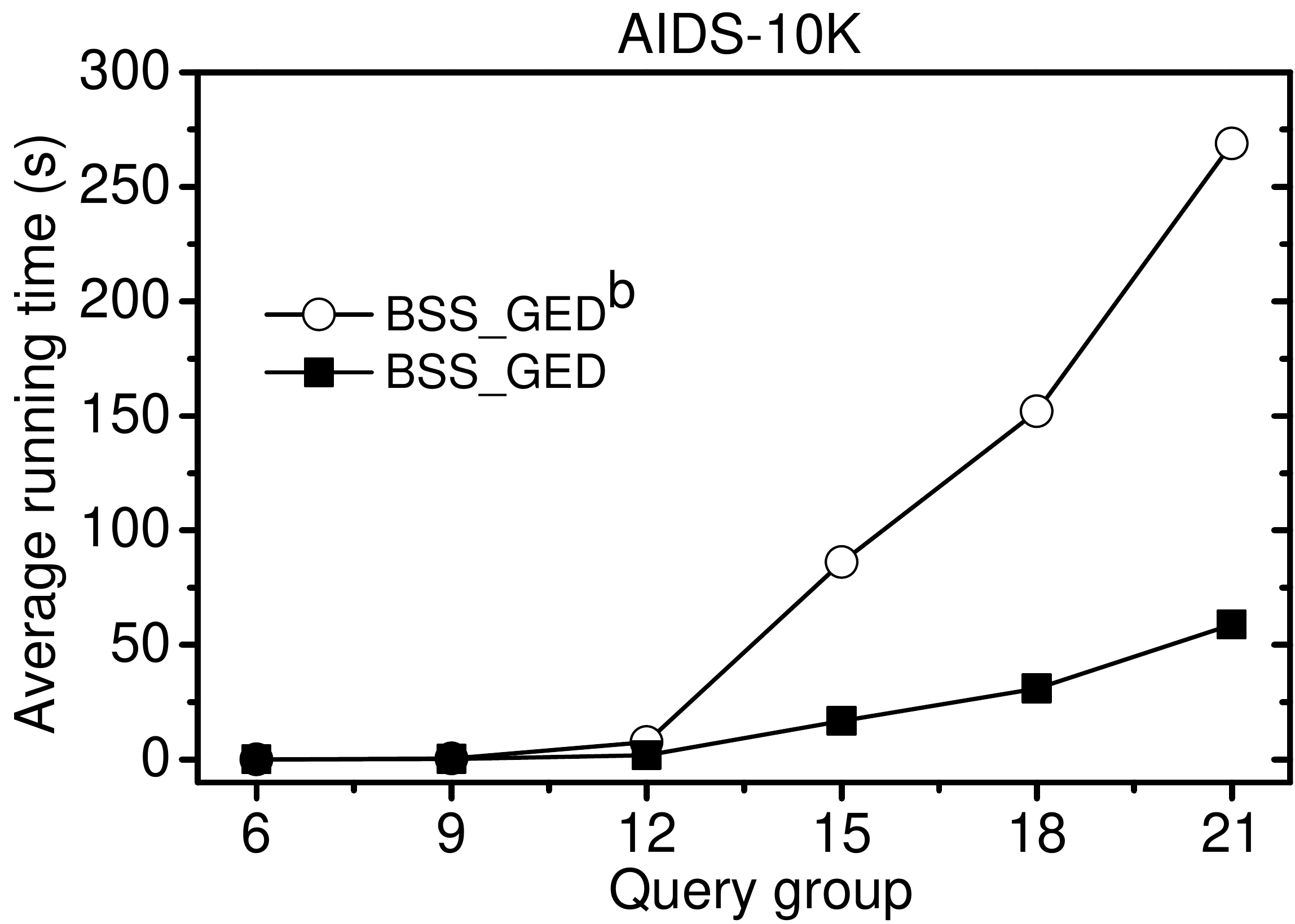}}
	\end{minipage}
	\qquad
	\begin{minipage}{0.28\linewidth}
		\centerline{\includegraphics[width=1\textwidth, height=3cm]{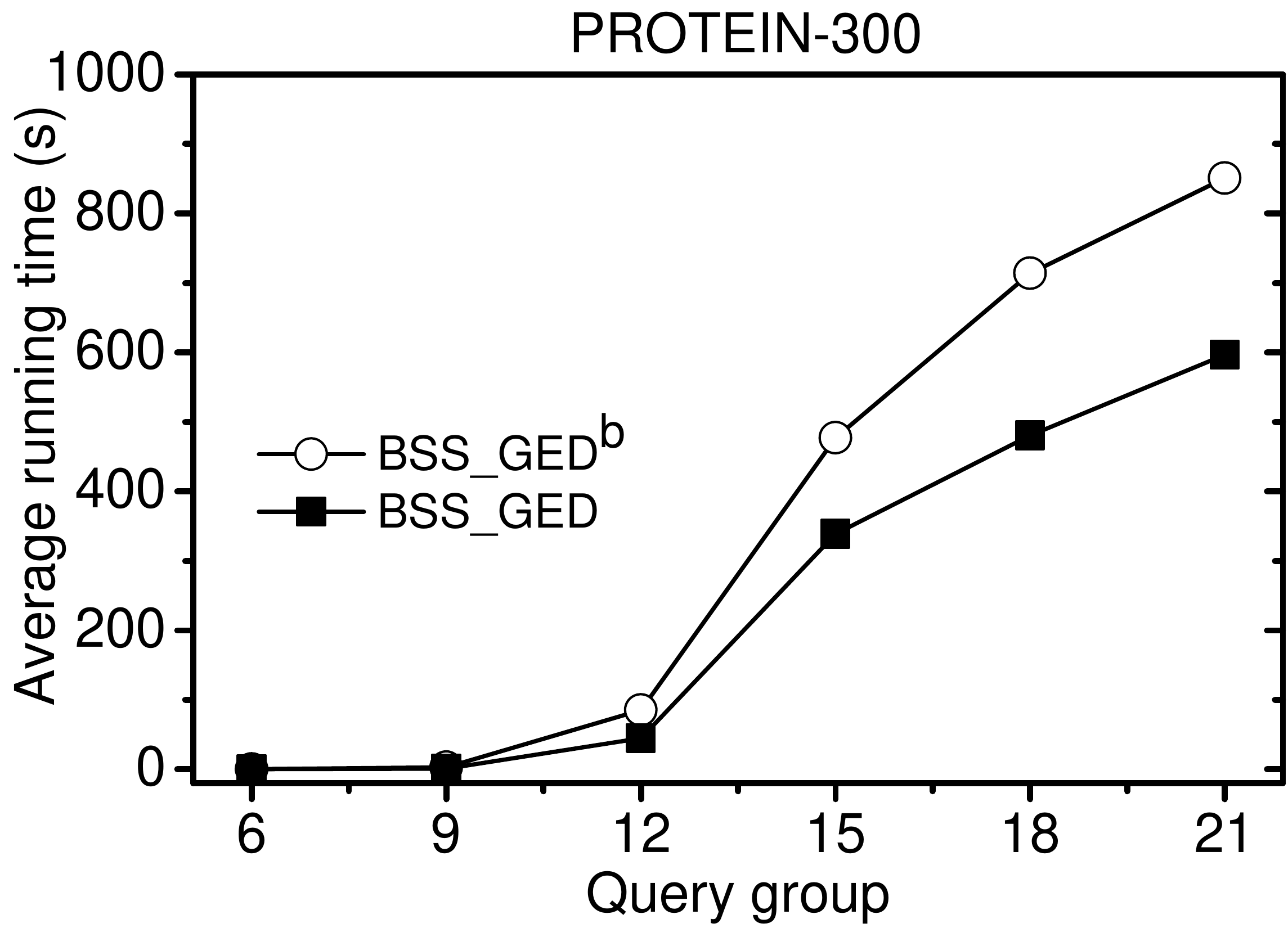}}
	\end{minipage}
	\qquad
	\begin{minipage}{0.28\linewidth}
		\centerline{\includegraphics[width=1\textwidth, height=3cm]{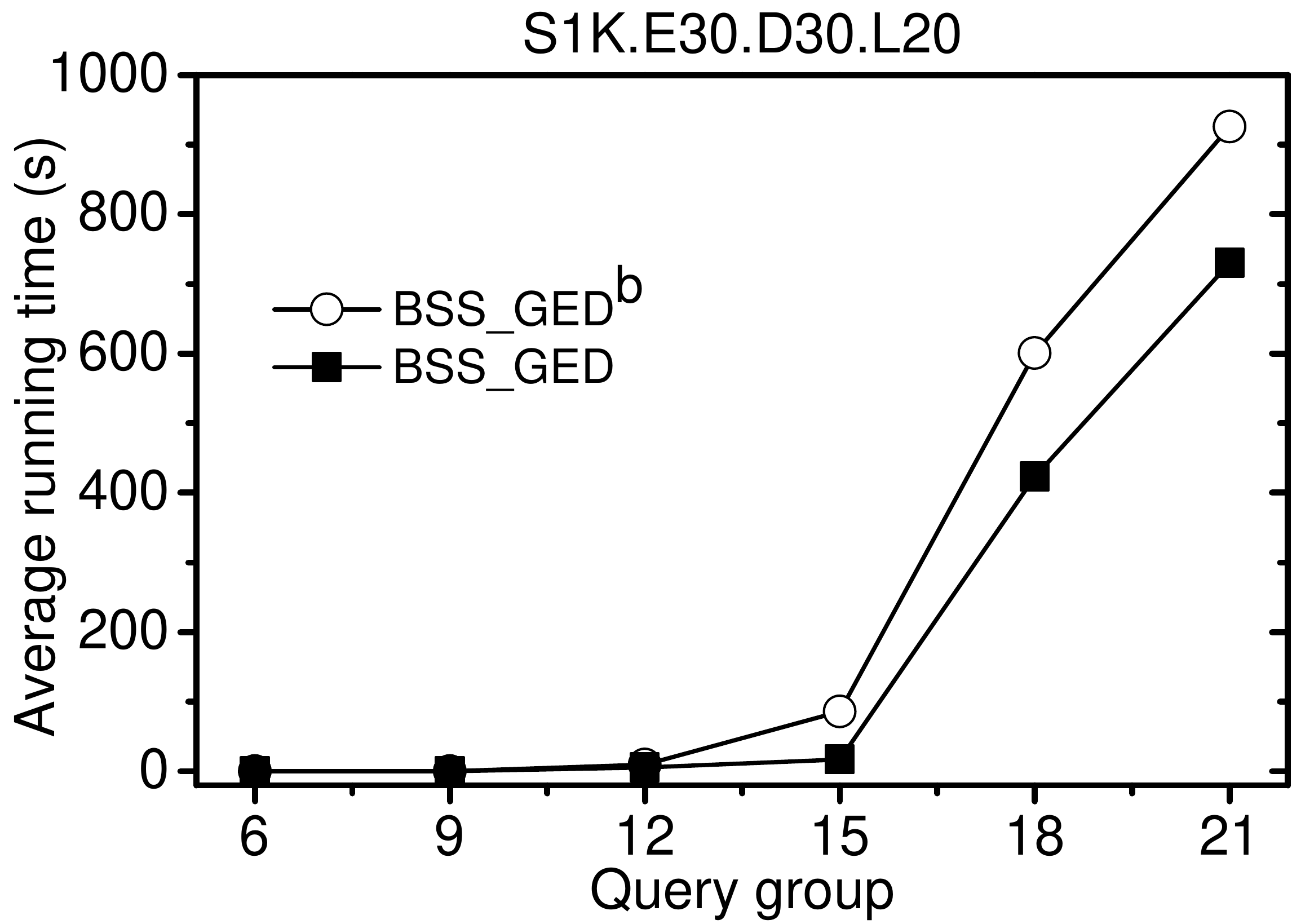}}
	\end{minipage}
    \caption{Effect of {\ttfamily GenSuccr} on the performance of {BSS\_GED}}
    \label{fig:Fig05}
\end{figure*}

\section{Experimental Results}
\label{sec:experiments}

In this section, we perform comprehensive experiments
and then analyse the obtained results.

\subsection{Datasets and Settings}
\label{subsec:Datasets}

We choose several real and synthetic datasets used in the experiment,
described as follows:

\begin{itemize}

\item AIDS\footnote{http://dtp.nci.nih.gov/docs/aids/aidsdata.html}.
It is an antivirus screen compound dataset from the Development
and Therapeutics Program in NCI/NIH, which contains 42687 chemical
compounds. We generate labeled graphs from these chemical compounds
and omit Hydrogen atom as did in~\cite{WilliamsHW2007,ChenHHJ2016}.

\item
PROTEIN\footnote{http://www.fki.inf.unibe.ch/databases/iam-graph-database/download-the-iam-graph-database}. It is a protein database from the Protein Data Bank, constituted of 600 protein structures.
Vertices represent secondary structure elements and are labeled with
their types (helix, sheet or loop). Edges are labeled to indicate if
two elements are neighbors or not.

\item Synthetic. The synthetic dataset is generated by the synthetic
graph data generator GraphGen\footnote{http://www.cse.ust.hk/graphgen/}.
In the experiment, we generate a density graph dataset S1K.E30.D30.L20,
which means that this dataset contains 1000 graphs;
the average number of edges in each graph is 30; the density\footnote{the
density of a graph $G$ is defined as $\frac{2|E_G|}{|V_G|(|V_G|-1)}$.}
of each graph is 30\%; and the distinct vertex and edge labels are
20 and~5, respectively.

\end{itemize}

Due to the hardness of computing GED, existing methods,
such as {A$^\star$-GED}~\cite{RiesenFB2007}, {DF-GED}~\cite{AbuRRM2015}
and {CSI\_GED}~\cite{GoudaH2016}, cannot obtain GED
of large graphs within a reasonable time and memory.
Therefore, for AIDS and PROTEIN, we exclude large graphs
with more than 30 vertices as did in~\cite{Blumenthal2017}, and
then randomly select 10000 and 300 graphs to make up the datasets
AIDS-10K and PROTEIN-300. For S1K.E30.D30.L20, we
use the entire dataset.

As suggested in~\cite{Blumenthal2017, GoudaH2016}, for each
dataset, we \mbox{randomly} se-lect~6 query groups, where each group
consists of~3 data graphs having three consecutive graph sizes.
{Specifically}, the number of vertices of each group
is in the range: $\idrm{6\pm1}, \idrm{9\pm1}, 12\pm1, 15\pm1, 18\pm1$ and $21\pm1$.

For the tested database $\mathcal{D} = \{\mathcal{D}_{1}, \mathcal{D}_{2}, \dots \}$
and query group $\mathcal{T} = \{\mathcal{T}_1, \mathcal{T}_{2}, \dots\}$,
we need to perform $|\mathcal{D}| \times |\mathcal{T}|$ times GED
computation. For each pair of the GED computation, we set the available time
and memory be 1000s and 24GB, respectively, and then define the metric \emph{average solve
ratio} as follows:
\begin{align}
\label{eq:sloveRate}
sr =\frac{\sum_{i=1}^{|\mathcal{D}|}\sum_{j=1}^{|\mathcal{T}|}
slove(\mathcal{D}_{i},\mathcal{T}_{j})}{|\mathcal{D}| \times |\mathcal{T}|}.
\end{align}
where $slove(\mathcal{D}_{i}, \mathcal{T}_{j}) = 1$ if we obtain
$ged(\mathcal{D}_{i}, \mathcal{T}_{j})$ within both 1000s and 24GB,
and $slove(\mathcal{D}_{i}, \mathcal{T}_{j}) = 0$ otherwise.
Obviously, $sr$ should be as large as possible.

We have conducted all experiments on a HP Z800 PC with a 2.67GHz
GPU and 24GB memory, running Ubuntu 12.04 operating system. We
implement our algorithm in C++, with -O3 to compile and run.
For {BSS\_GED}, we set the beam width $\id{w = \idrm{15}}$ for
the sparse graphs in datasets AIDS-10K and PROTEIN-300, and
$\id{w=\idrm{50}}$ for the density graphs in dataset S1K.E30.D30.L20.

\subsection{Evaluating GenSuccr}
\label{subsec:evalutingGS}

In this section, we evaluate the effect of {\ttfamily GenSuccr} on the
performance of {BSS\_GED}. To make a comparison, we replace
{\ttfamily GenSuccr} with {\ttfamily BasicGenSuccr} (i.e., Alg.~\ref{alg:general})
and then obtain {BSS\_GED}$^\idrm{b}$, where {\ttfamily BasicGenSuccr}
is the basic method of generating successors
used in {A$^\star$-GED}~\cite{RiesenFB2007} and {DF-GED}~\cite{AbuRRM2015}.
In {BSS\_GED}$^\idrm{b}$, we also use the same heuristics proposed
in Section~\ref{sec:optimization}. Figure~\ref{fig:Fig05} shows the
average solve ratio and running time.

As shown in Figure~\ref{fig:Fig05}, the average solve ratio of
\mbox{BSS\_GED} is much higher than that of {BSS\_GED}$^\idrm{b}$, and
the gap between them becomes larger as the query graph size increases.
This indicates that {\ttfamily GenSuccr} provides more reduction on
the search space for larger graphs. Regarding the running time,
{BSS\_GED} achieves the respective 1x--5x, 0.4x--1.5x and \mbox{0.1x--4x} speedup
over {BSS\_GED}$^\idrm{b}$ on AIDS-10K, \mbox{PROTEIN-300} and
S1K.E30.D50.L20. Thus, we create a small search space by {\ttfamily GenSuccr}.

\subsection{Evaluating BSS\_GED}
\label{subsec:EvaluatingBSS}

In this section, we evaluate the effect of beam-stack search and
heuristics on the performance of {BSS\_GED}. We fix datasets AIDS-10K,
PROTEIN-300 and S1K.E30.D30.L20 as the tested datasets and select
their corresponding groups $15\pm1$ as the query groups, respectively.

\noindent (1). \emph{Effect of} $w$

As we know, beam-stack search achieves a flexible trade-off between available
memory and expensive backtracking by setting \mbox{different $w$}, thus we
vary~$w$ to evaluate its effect on the performance. Figure~\ref{fig:Fig06}
shows the average solve ratio and running time.

\begin{figure}[htbp]
\centering
	\begin{minipage}{0.47\linewidth}
		\centerline{\includegraphics[width=1\textwidth, height=2.7cm]{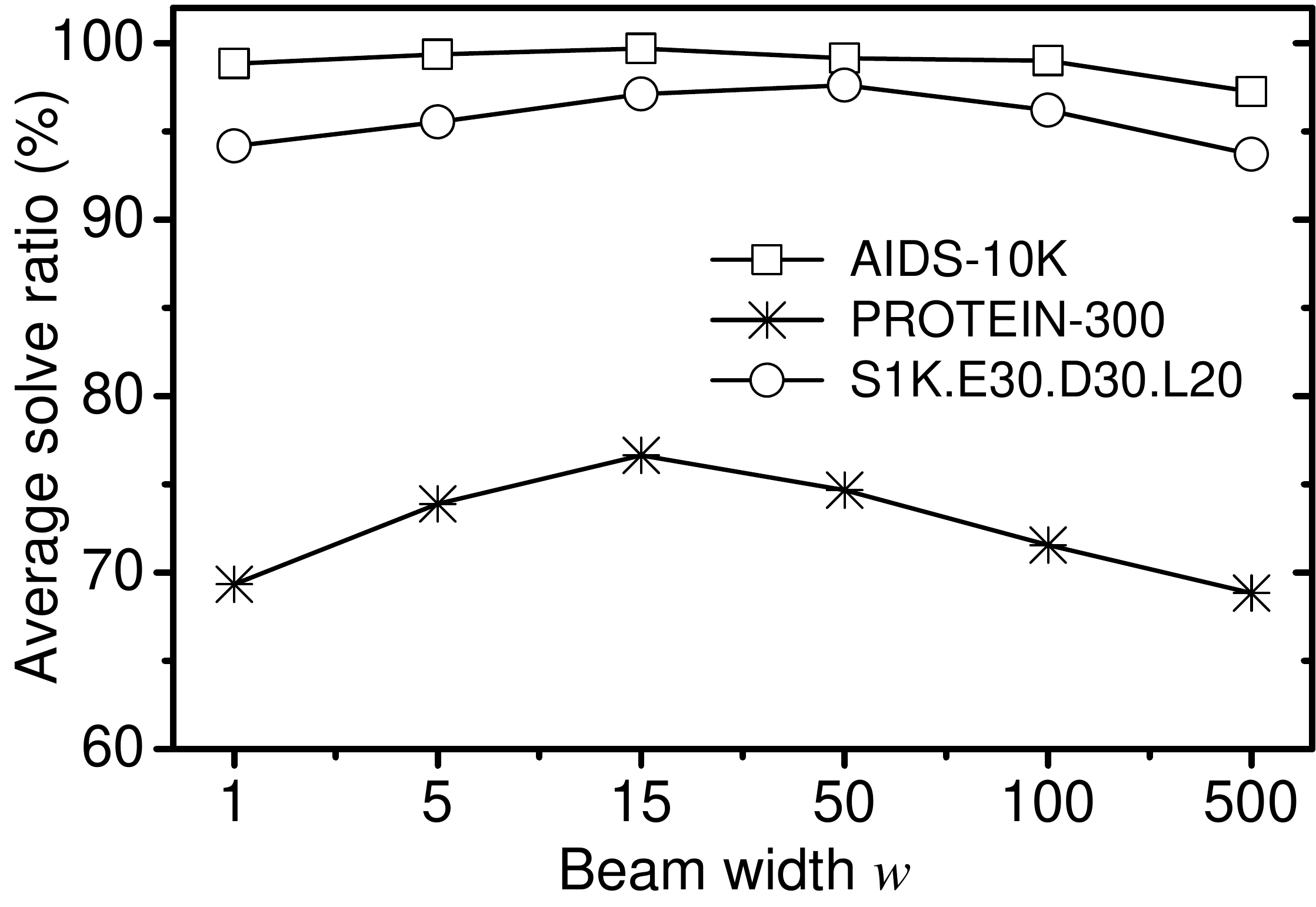}}
	\end{minipage}
	\quad
	\begin{minipage}{0.47\linewidth}
		\centerline{\includegraphics[width=1\textwidth, height=2.7cm]{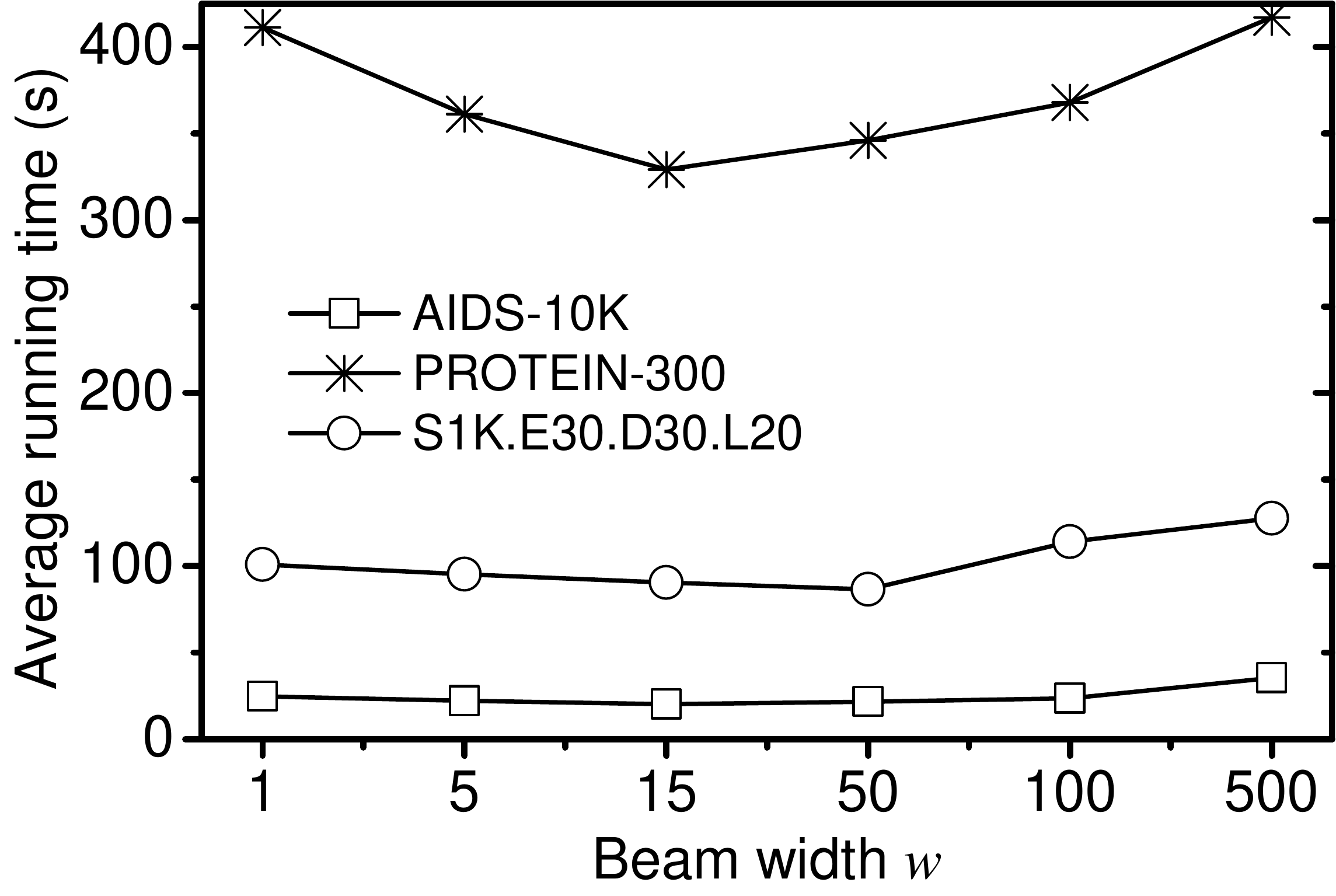}}
	\end{minipage}
    \caption{Effect of $w$ on the performance of {BSS\_GED}.}
	\vspace{-5pt}	
    \label{fig:Fig06}
\end{figure}

By Figure~\ref{fig:Fig06}, we obtain that the average
solve ratio first increases and then decreases, and achieves maximum
when $\id{w = \idrm{15}}$ on AIDS-10K and PROTEIN-300, and $\id{w = \idrm{50}}$
on S1K.E30.D30.L20. There are several factors may contribute to this trend:
(1) When $w$ is too small, {BSS\_GED} may be trapped into
a local suboptimal solution and hence produces lots of backtracking.
(2) When $w$ is too large, {BSS\_GED} expands too many
unnecessary nodes in each layer. Note that, depth-first search is
a special case of beam-stack search when $\id{w = \idrm{1}}$. Thus,
beam-stack search performs better than depth-first search. As
previously demonstrated in~\cite{AbuRRM2015}, depth-first search
performs better than best-first search. Therefore, we conclude that
the beam-stack search paradigm outperforms the best-first and depth-first
search paradigms for the GED computation.
\begin{figure}[htbp]
\centering
	\begin{minipage}{0.47\linewidth}
		\centerline{\includegraphics[width=1\textwidth, height=2.6cm]{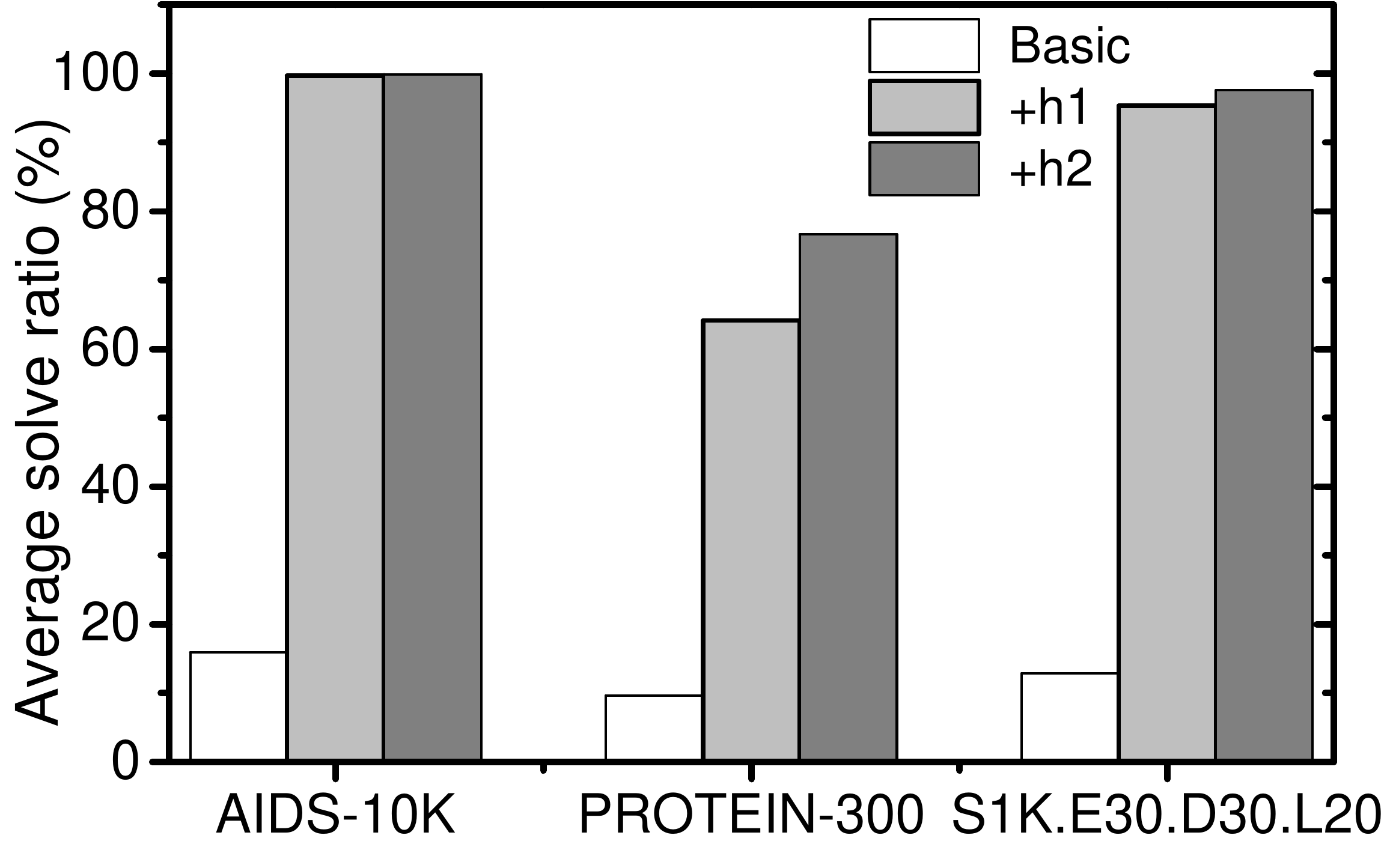}}
	\end{minipage}
	\quad
	\begin{minipage}{0.47\linewidth}
		\centerline{\includegraphics[width=1\textwidth, height=2.7cm]{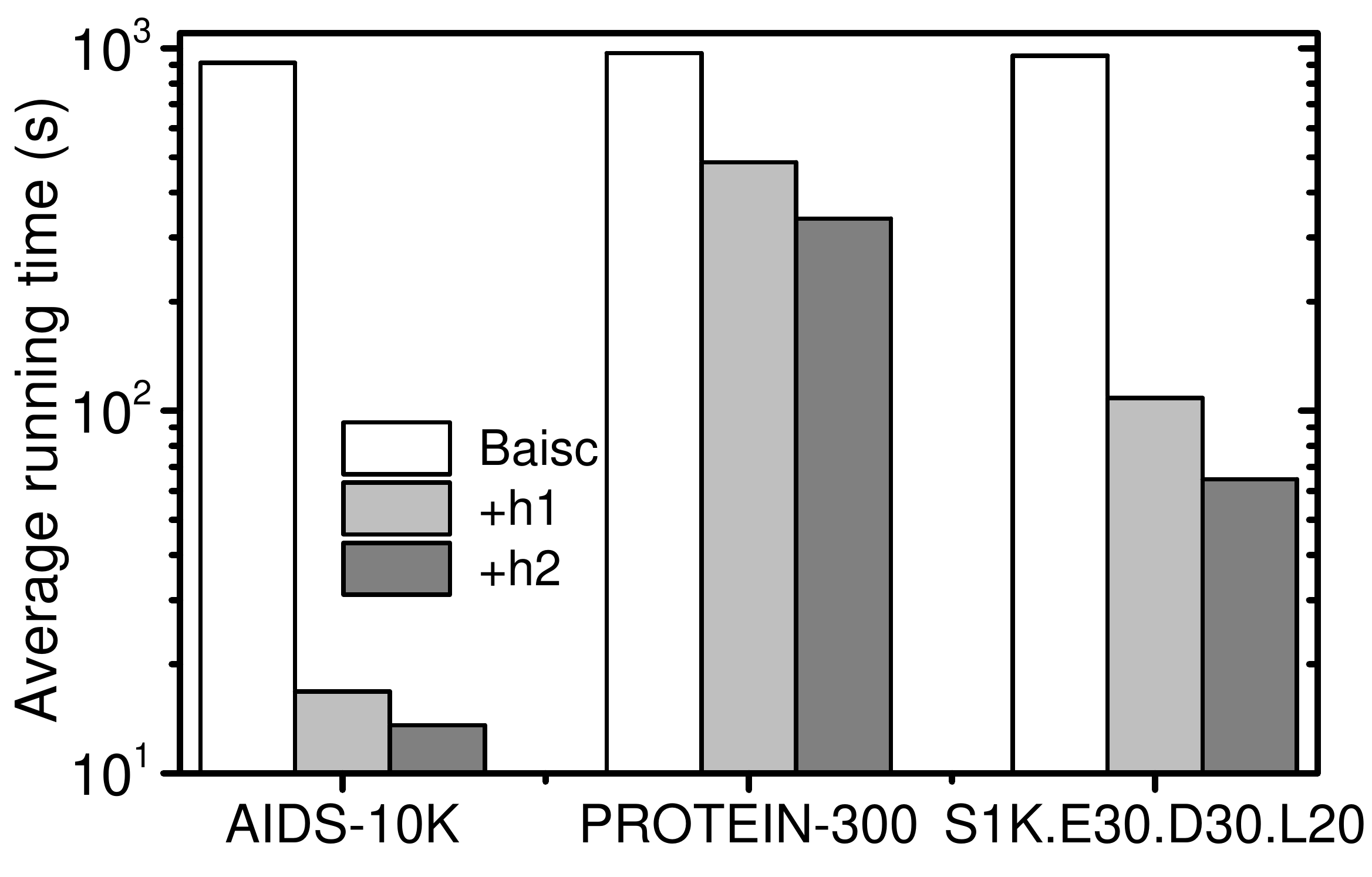}}
	\end{minipage}
\caption{Effect of heuristics on the performance of {BSS\_GED}.}
\vspace{-5pt}
\label{fig:Fig07}
\end{figure}
\begin{figure*}[htbp]
\centering
 \begin{minipage}{0.28\linewidth}
		\centerline{\includegraphics[width=1\textwidth, height=3cm]{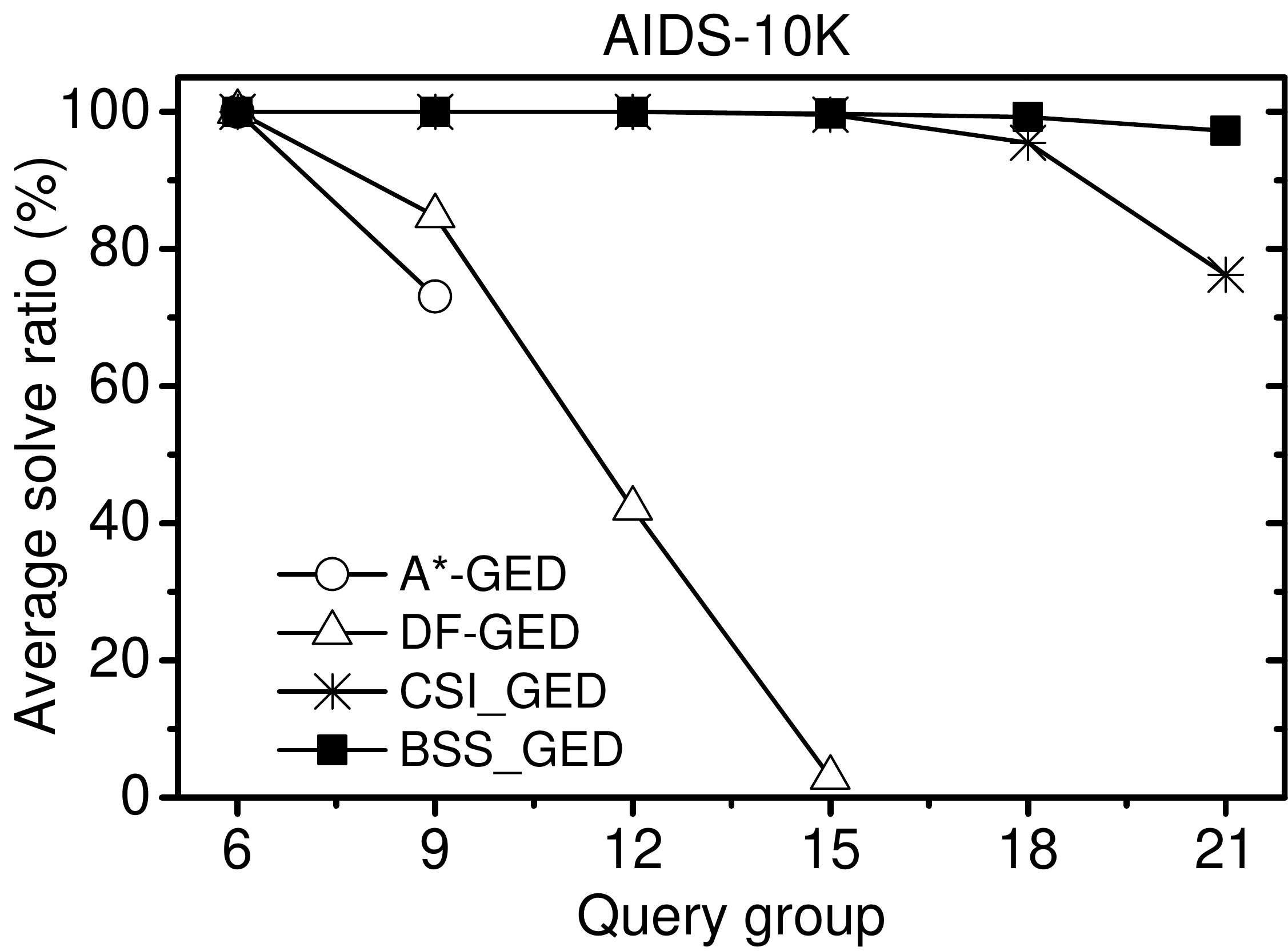}}
	\end{minipage}
	\qquad
	\begin{minipage}{0.28\linewidth}
		\centerline{\includegraphics[width=1\textwidth, height=3cm]{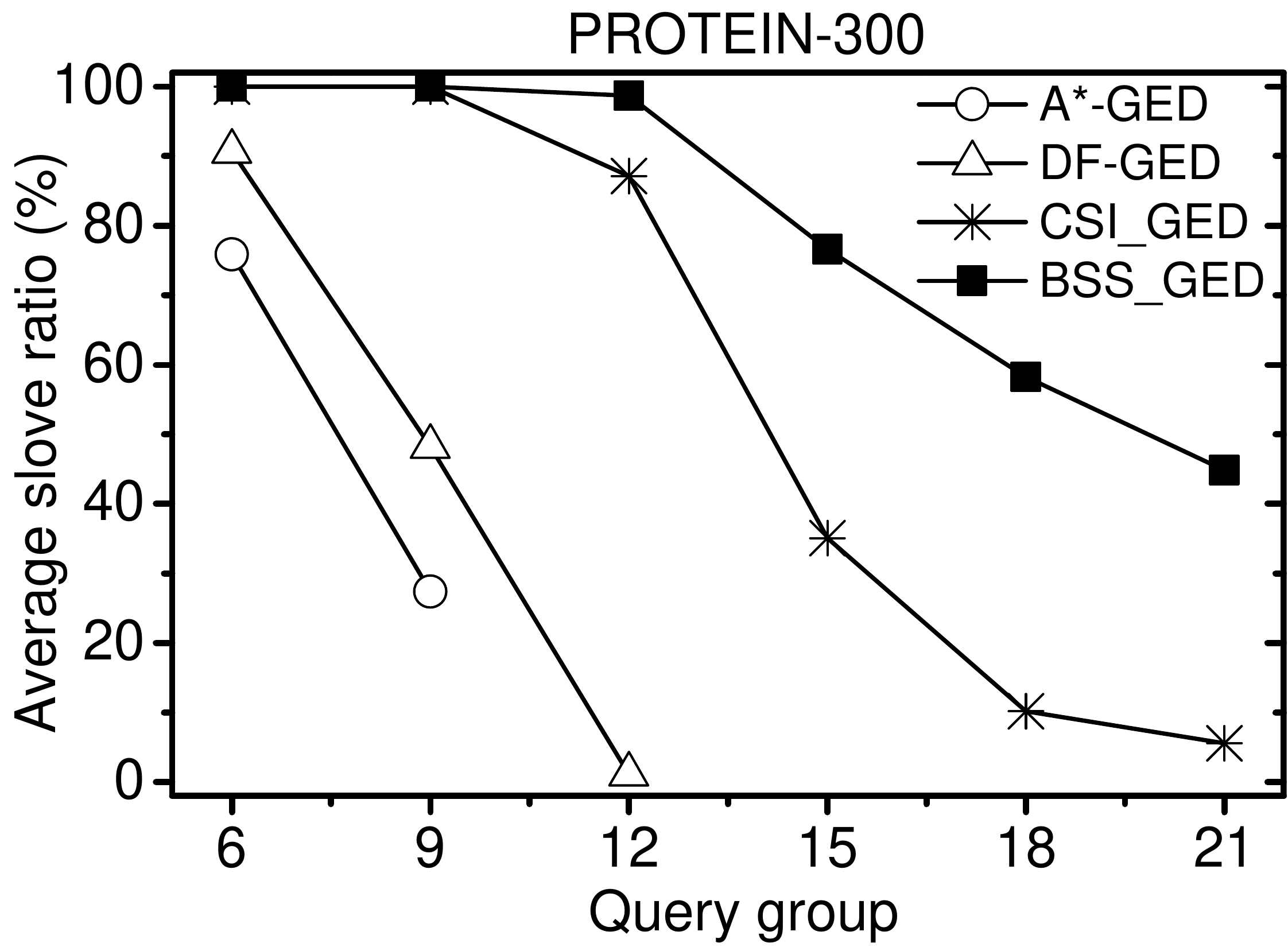}}
	\end{minipage}
	\qquad
	\begin{minipage}{0.28\linewidth}
		\centerline{\includegraphics[width=1\textwidth, height=3cm]{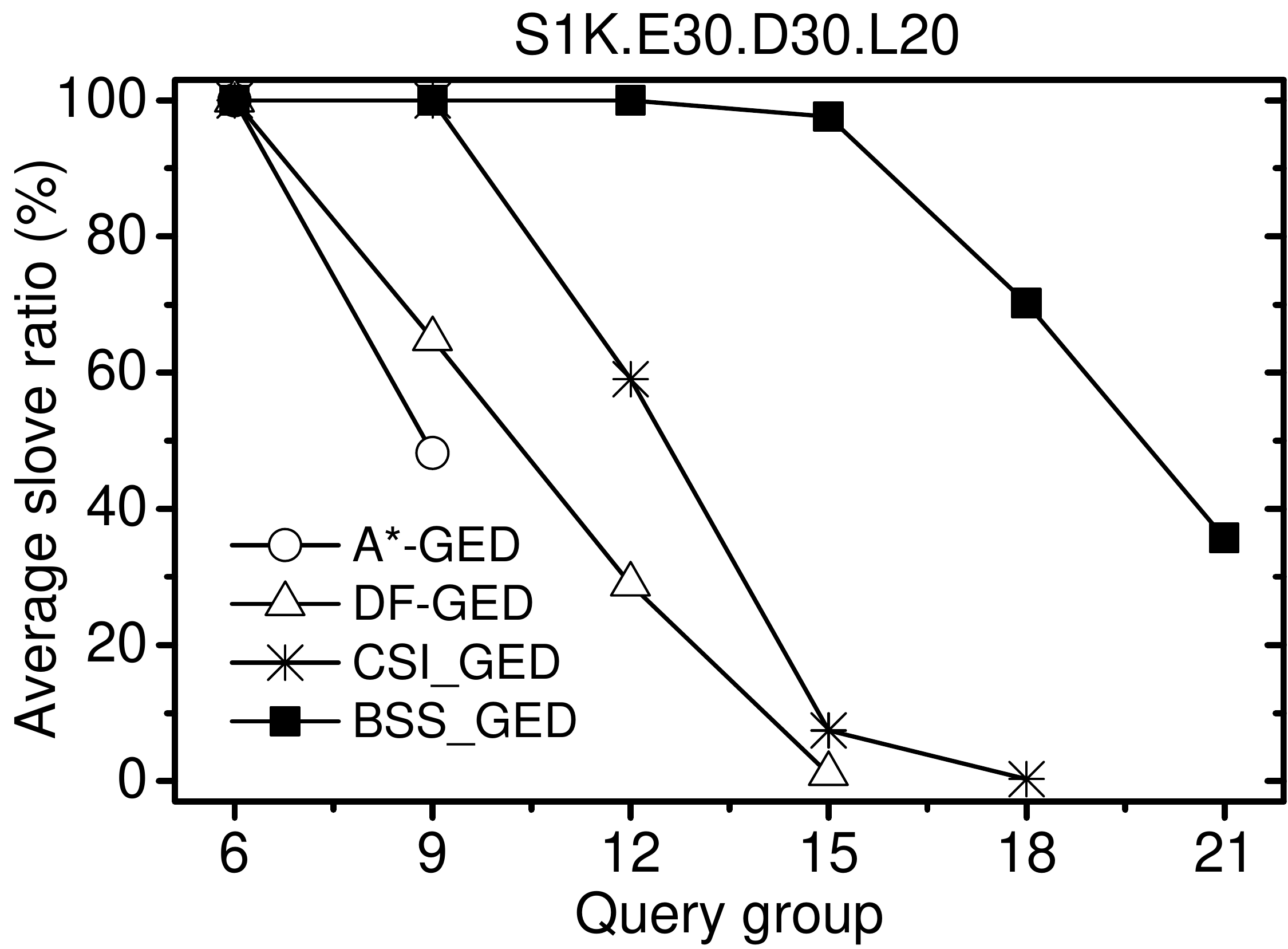}}
    \end{minipage}
    \vfill
    \begin{minipage}{0.28\linewidth}
		\centerline{\includegraphics[width=1\textwidth, height=3cm]{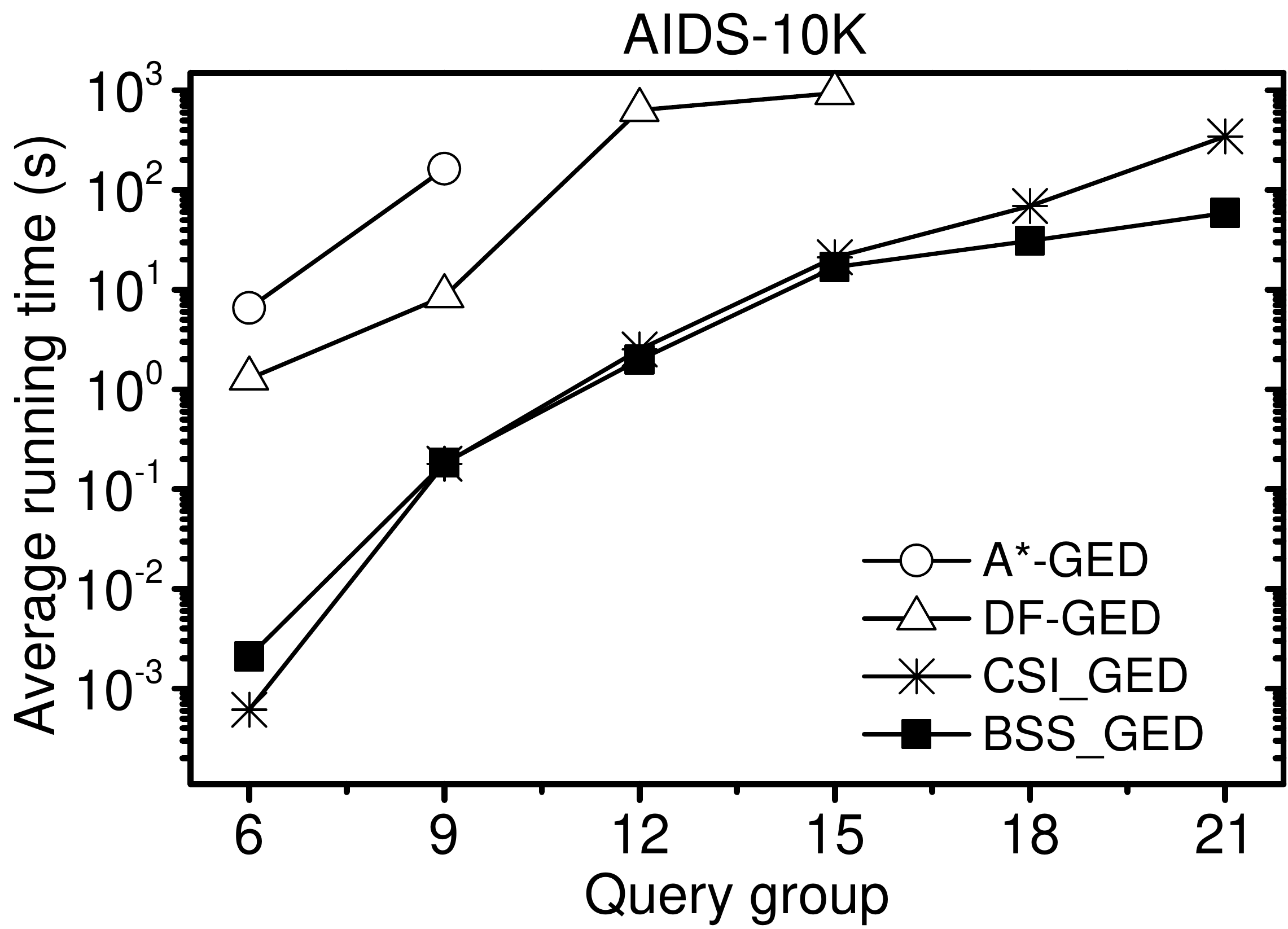}}
	\end{minipage}
	\qquad
	\begin{minipage}{0.28\linewidth}
		\centerline{\includegraphics[width=1\textwidth, height=3cm]{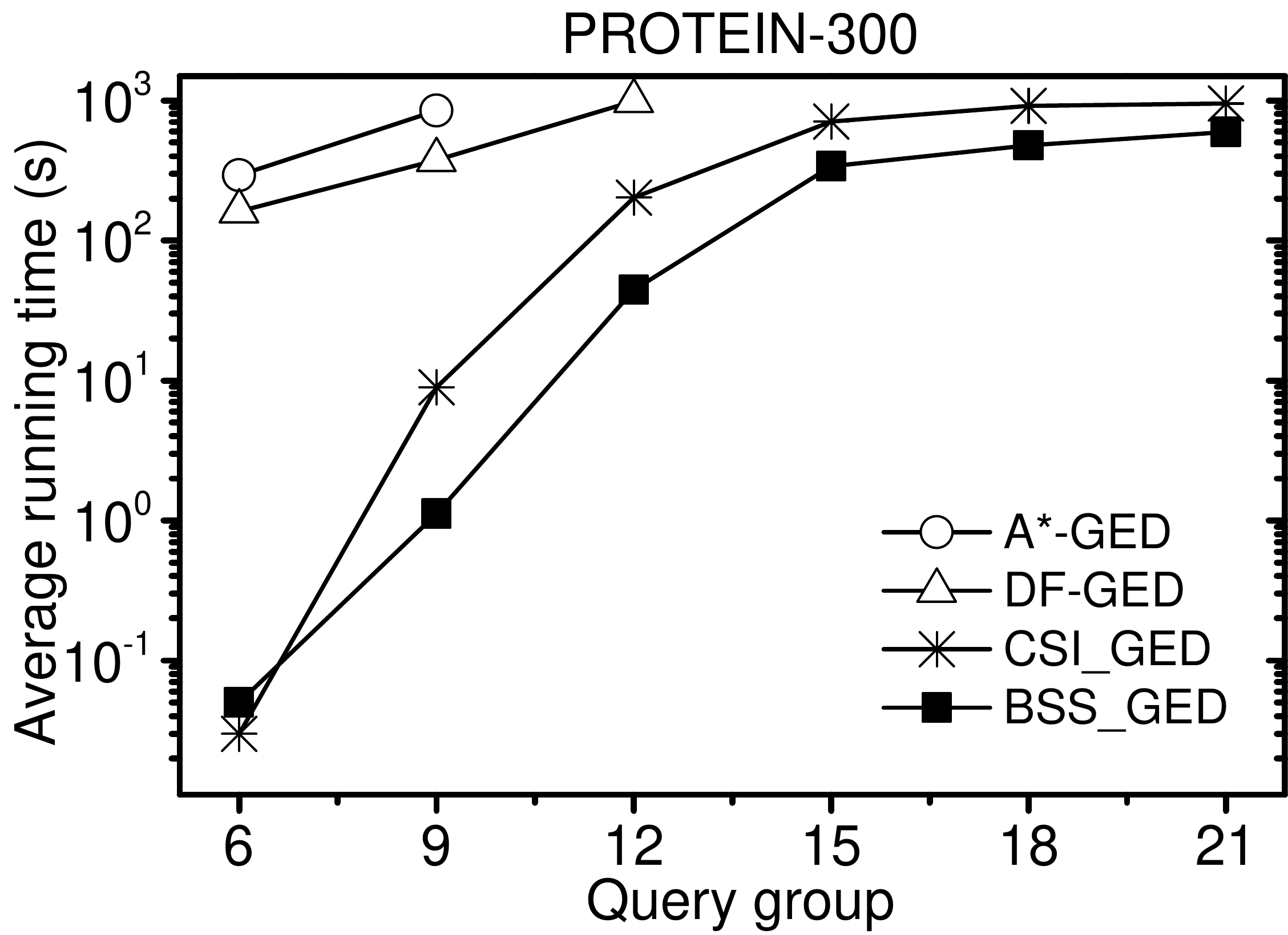}}
	\end{minipage}
	\qquad
	\begin{minipage}{0.28\linewidth}
		\centerline{\includegraphics[width=1\textwidth, height=3cm]{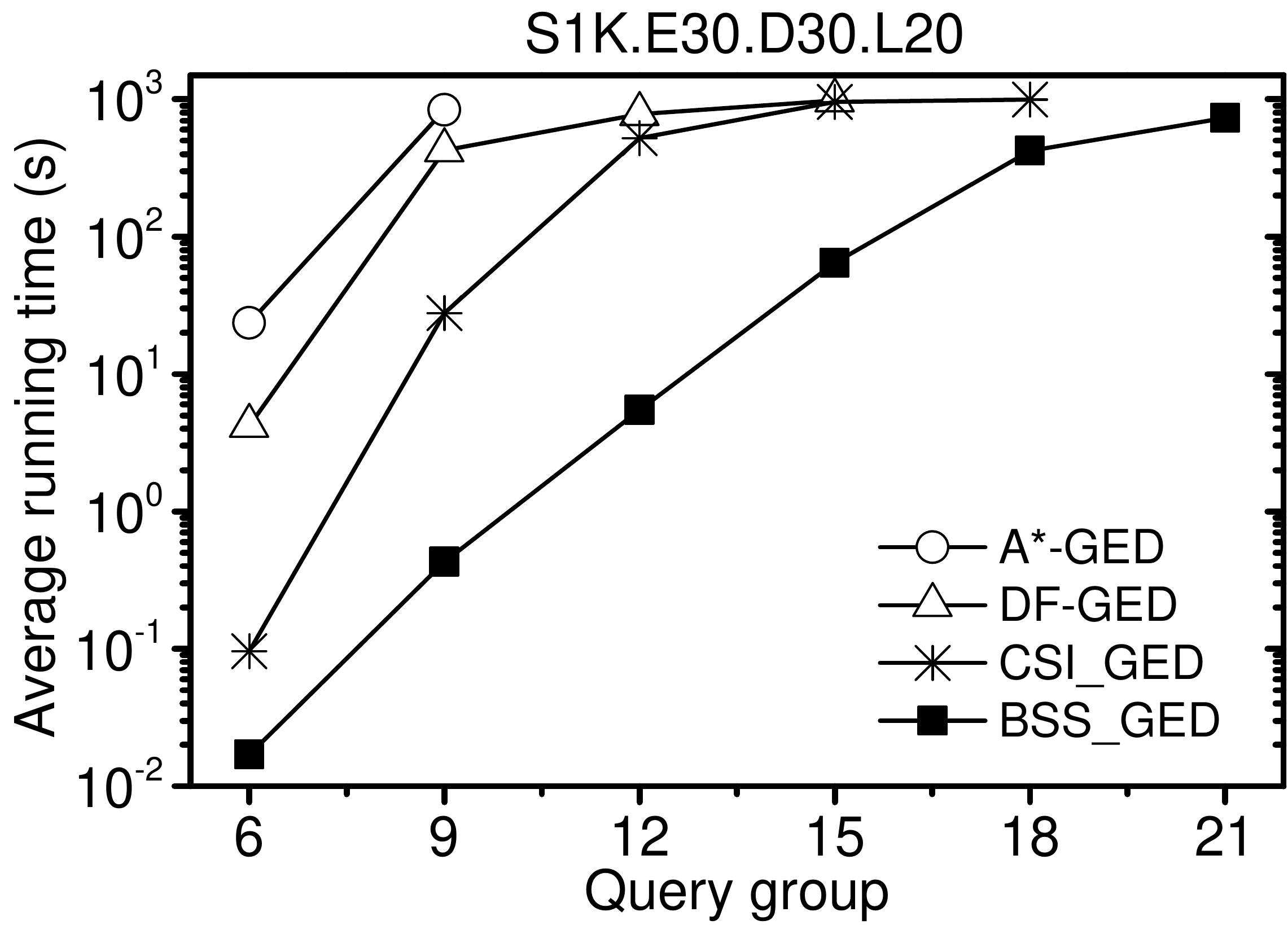}}
	\end{minipage}
    \caption{Performance comparison with existing state-of-the-art GED methods.}
    \label{fig:Fig08}
\end{figure*}

\noindent (2). \emph{Effect of Heuristics}

In this part, we evaluate the effect of the proposed two heuristics
by injecting them one by one into the base algo-rithm.
We use the term {\ttfamily Basic} for the baseline algorithm without
applying any heuristics. {\ttfamily +h1} denotes the improved algorithm
of {\ttfamily Basic} by incorporating the first heuristics (Section~\ref{subsec:lowerBounds}). {\ttfamily +h2} denotes the improved algorithm of {\ttfamily +h1} by
incorporating the second heuristic (Section~\ref{subsec:orderingVertices}).
Figure~\ref{fig:Fig07} plots the average solve ratio and running time.

By Figure~\ref{fig:Fig07}, we know that the average solve ratio
of {\ttfamily Basic} is only 15\% of that of {\ttfamily +h1}.
This means that the proposed heuristic function provides powerful
pruning ability. Considering the running time, {\ttfamily +h1}
brings the respective 50x, 2x and 9x speedup over
{\ttfamily Basic} on AIDS-10K, PROTEIN-300 and S1K.E30.D30.L20. Moreover,
compared with {\ttfamily +h1}, the running time needed by
{\ttfamily +h2} decreases 21\%, 30\% and 41\% on AIDS-10K, PROTEIN-300
and S1K.E30.D30.L20, respectively. Thus, the proposed two heuristics
greatly boost the performance.

\subsection{Comparing with Existing GED Methods}

In this section, we compare {BSS\_GED} with existing \mbox{methods}
{A$^\star$-GED}~\cite{RiesenFB2007}, {DF-GED}~\cite{AbuRRM2015}
and {CSI\_GED}~\cite{GoudaH2016}. Figure~\ref{fig:Fig08} shows
the average solve ratio and running time.

By Figure~\ref{fig:Fig08}, we know that {BSS\_GED}
performs the best in terms of average solve ratio. For {A$^\star$-GED},
it cannot obtain GED of graphs with more than 12 vertices
within~24GB. For {DF-GED}, it cannot finish the GED
computation of graphs with more than 15 vertices in 1000s. Besides,
for density graphs in S1K.E30.D30.L20, the average solve ratio of
{CSI\_GED} drops sharply as the query graph size increases, which
confirms that it is unsuitable for dense graphs.

Regarding the running time, {BSS\_GED} still performs the best
in most cases. {DF-GED} performs better than {A$^\star$-GED}, which is
consistent with the previous results in~\cite{AbuRRM2015}. Compared with {DF-GED},
{BSS\_GED} achieves 50x--500x, 20x--2000x and 15x--1000x speedup on
AIDS-10K, PROTEIN-300 and S1K.E30.D30.L20, respectively.
Though {CSI\_GED} performs better than {BSS\_GED} on AIDS-10K when
the query graph size is less than 9, {BSS\_GED} achieves 2x--5x speedup over
{CSI\_GED} when the graph size is greater than 12. Besides, for S1K.E30.D30.L20,
{BSS\_GED} achieves 5x--95x speedup over {CSI\_GED}. Thus, {BSS\_GED}
is efficient for the GED computation on sparse as well as dense graphs.

\subsection{Performance Evaluation on Graph Similarity Search}

In this part, we evaluate the performance of {BSS\_GED} as a
standard graph similarity search query method by comparing it with
{CSI\_GED} and {GSimJoin}~\cite{ZhaoXLW2012}.
For each dataset described in Section~\ref{subsec:Datasets}, we
use its entire dataset and randomly select 100 graphs from it
as query graphs. Figure~\ref{fig:Fig09} shows the total running
time (i.e., the filtering time plus the verification time).

It is clear from Figure~\ref{fig:Fig09} that {BSS\_GED} has the best
performance in most cases, especially for large~$\tau$. For {GSimJoin},
it cannot finish when $\id{\tau \geq \idrm{8}}$ in AIDS and PROTEIN because of the
huge memory consumption. Compared with {GSimJoin} for $\tau$ values where it can
finish, {BSS\_GED} achieves the respective 1.6x--15000x, 3.8x--800x and 2x--3000x
speedup on \mbox{AIDS}, PROTEIN and S1K.E30.D30.L20.
Considering {CSI\_GED}, it performs slightly better than \mbox{BSS\_GED} when
$\id{\tau \leq \idrm{4}}$ on AIDS. However, {BSS\_GED} performs much
better than {CSI\_GED} when $\id{\tau \geq \idrm{6}}$ and the gap between them becomes
larger as~$\tau$ increases. Specifically, {BSS\_GED} achieves
2x--28x, 1.2x--100000x and 1.1x--187x speedup over {CSI\_GED} on AIDS,
PROTEIN and S1K.E30.D30.L20, respectively. As previously discussed in~\cite{GoudaH2016},
{CSI\_GED} performs much better than the state-of-the-art graph similarity
search query methods. Thus, we conclude that {BSS\_GED} can efficiently finish the graph
similarity search and runs much faster than existing methods.
\begin{figure*}[!htbp]
\centering
\begin{minipage}{0.28\linewidth}
	\centerline{\includegraphics[width=1\textwidth, height=3cm]{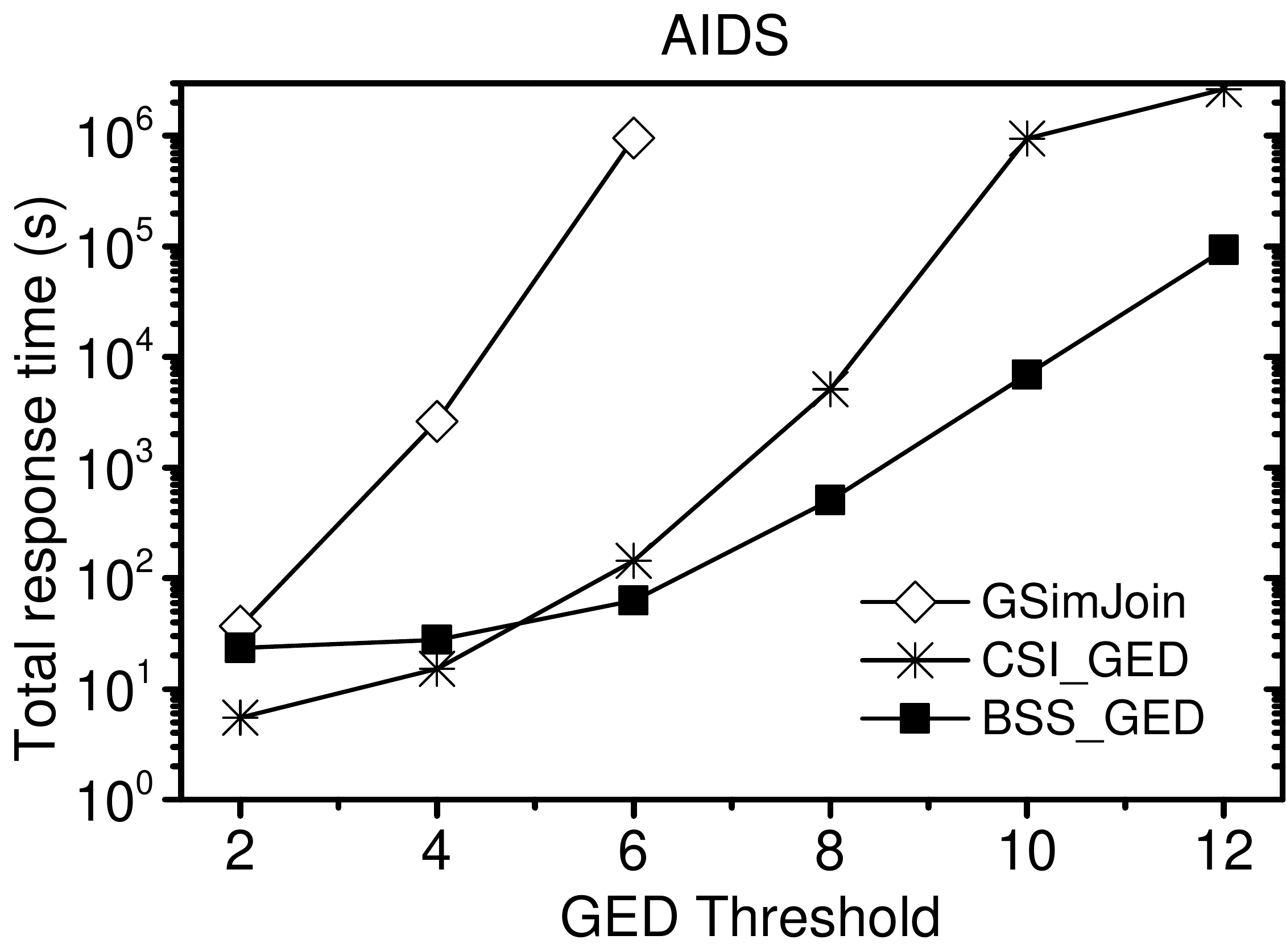}}
\end{minipage}
\qquad
\begin{minipage}{0.28\linewidth}
	\centerline{\includegraphics[width=1\textwidth, height=3cm]{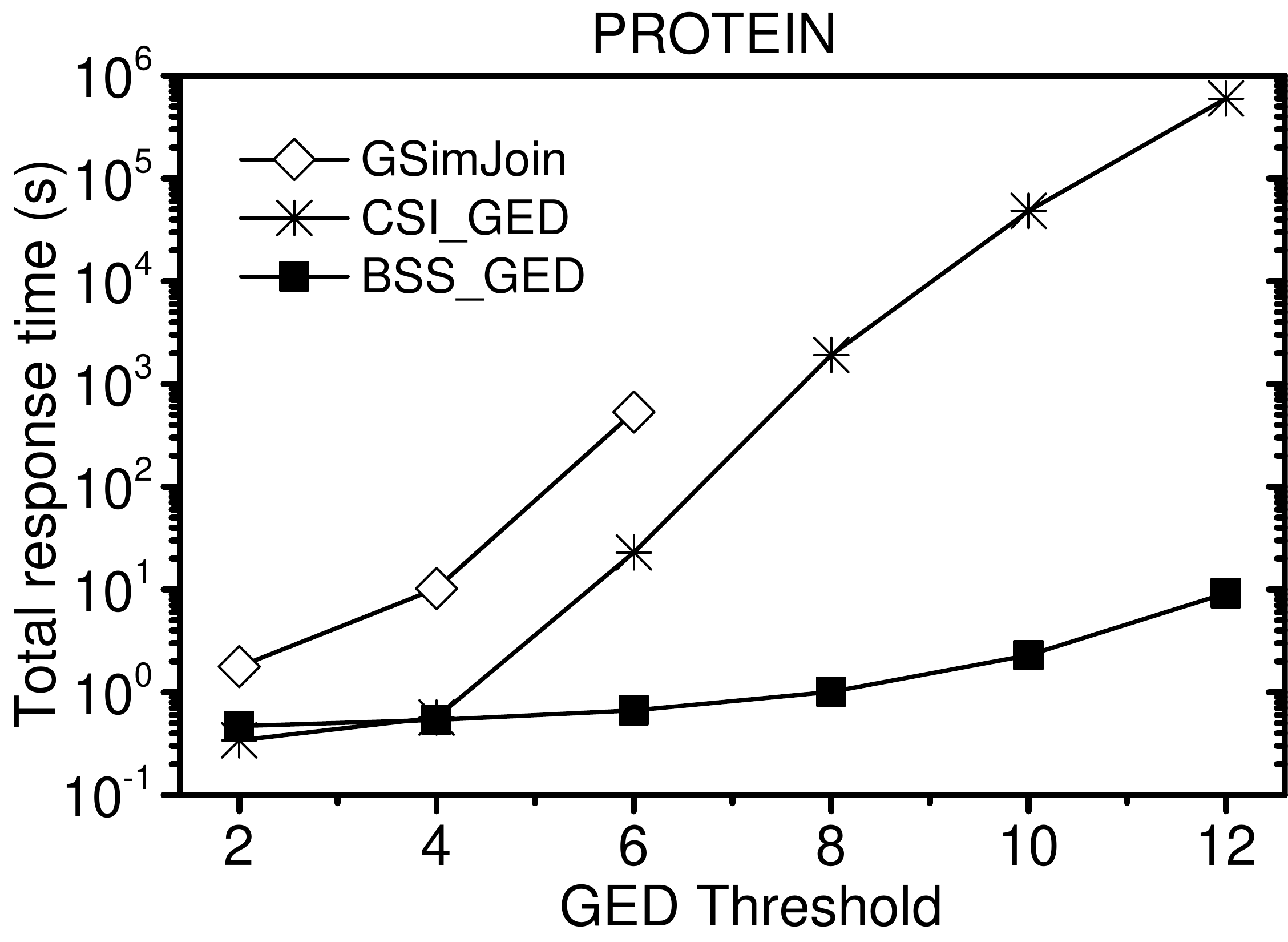}}
\end{minipage}
\qquad
\begin{minipage}{0.28\linewidth}
	\centerline{\includegraphics[width=1\textwidth, height=3cm]{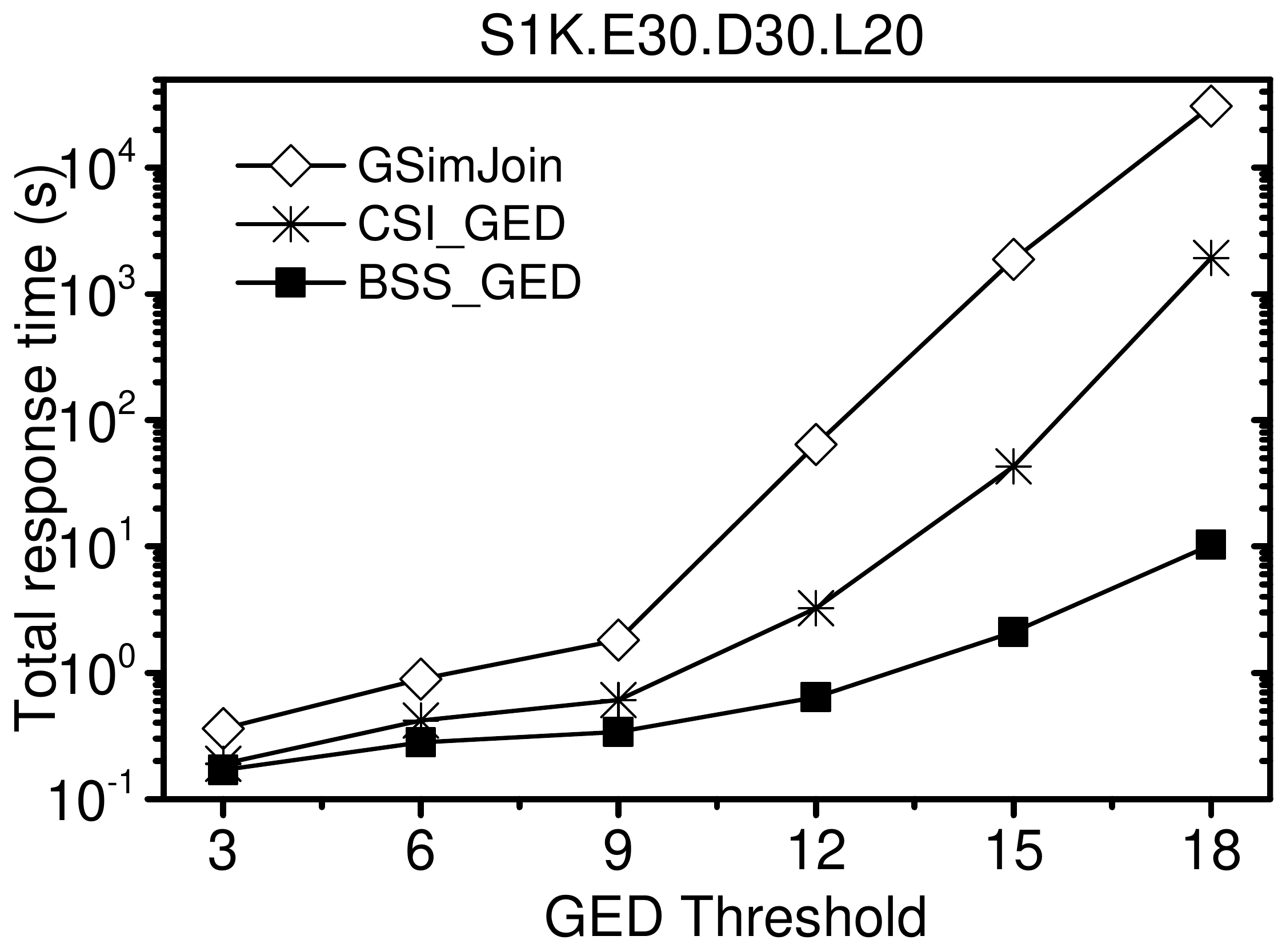}}
\end{minipage}
\caption{Performance comparison with CSI\_GED and GSimJoin.}
\label{fig:Fig09}
\end{figure*}

\section{Related Works}
\label{sec:relatedWorks}

Recently, the GED computation has received considerable attention.
{A$^\star$-GED}~\cite{RiesenFB2007, RiesenEB2013} and {DF-GED}~\cite{AbuRRM2015}
are two major vertex-based mapping methods, which utilize the
best-first and depth-first search paradigms, respectively.
Provided that the heuristic function estimates the lower bound of
GED of unmapped parts, {A$^\star$-GED} guarantees that the first found
complete mapping induces the GED of comparing graphs, which seems very
attractive. However, {A$^\star$-GED} stores numerous partial
mappings, resulting in a huge memory consumption. As a result,
it is only suitable for the small graphs. To overcome this bottleneck,
{DF-GED} performs a depth-first search, which only stores the partial
mappings of a path from the root to a leaf node. However, it may
be easily trapped into a local suboptimal solution, leading to massive expensive
backtracking. {CSI\_GED}~\cite{GoudaH2016} is an edge-based mapping method
through common substructure isomorphism enumeration, which has an
excellent performance on the sparse graphs. However, the edge-based
search space of {CSI\_GED} is exponential with respect to the number
of edges of comparing graphs, making it naturally be unsuitable for
dense graphs. Note that, {CSI\_GED} only works for the uniform model,
and~\cite{Blumenthal2017} generalized it to cover the non-uniform model.

Another work closely related to this paper is the GED based graph similarity
search problem. Due to the hardness of computing GED,
existing graph similarity search query methods~\cite{WangWYY2012, ChenHHJ2016, ZengTWFZ2009, ZhaoXLLZ2013, ZhaoXLW2012, ZhengZLWZ2015} all adopt the filter-and-verify
schema, that is, first filtering graphs to obtain a candidate set
and then verifying those candidate graphs. In the verification phase,
most of the existing methods adopt {A$^\star$-GED} as their verifiers.
As discussed above, {BSS\_GED} greatly outperforms {A$^\star$-GED},
hence it can be also used as a standard verifier to accelerate
those graph similarity search query methods.

\section{Conclusions and Future Work}
\label{sec:conclusion}

In this paper, we present a novel vertex-based mapping method
for the GED computation. First, we reduce the number of invalid and
redundant mappings involved in the GED computation and then create
a small search space. Then, we utilize beam-stack search to
efficiently traverse the search space to compute GED, achieving
a flexible trade-off between available memory and expensive backtracking.
In addition, we also give two efficient heuristics to prune the search space.
However, it is still very hard to compute GED of large graphs within a reasonable
time. Thus, the approximate algorithm that fast suboptimally compute GED 
is left as a future work.

\section{Acknowledgments}
The authors would like to thank Kaspar Riesen and Zeina Abu-Aisheh
for providing their source files, and thank Karam Gouda and Xiang Zhao
for providing their executable files. This work is supported in part
by China NSF grants 61173025 and 61373044, and US NSF grant CCF-1017623.
Hongwei Huo is the corresponding author.

\end{document}